\documentclass[10pt,prd,longbibliography,superscriptaddress]{revtex4-2}
\usepackage{graphicx}
\usepackage{epstopdf,subfigure}

\usepackage{amsthm}
\usepackage{amsfonts, amssymb, amsmath, amscd, latexsym,indentfirst}

\usepackage[figuresright]{rotating}

\newtheorem{theorem}{Theorem}
\newtheorem{lemma}{Lemma}

\newtheorem{remark}{Remark}
\newtheorem{example}{Example}

\def\ga{\gamma}
\def\de{\delta}

\def\ep{\varepsilon}

\def\ka{\kappa}
\def\la{\lambda}
\def\om{\omega}

\def\de{\delta}
\def\C{\mathbb{C}}
\def\I{\mathbb{I}}
\def\N{\mathbb{N}}
\def\Q{\mathbb{Q}}
\def\R{\mathbb{R}}
\def\Z{\mathbb{Z}}
\def\FF{{\mathcal F}}
\def\GG{{\mathcal G}}

\def\KK{{\mathcal K}}

\def\NN{{\mathcal N}}
\def\RR{{\mathcal R}}
\def\({\left(}
\def\){\right)}

\def\du{{\mathrm D}}
\def\wt{\widetilde}

\def\eu{ {\, \textrm{\rm e} }}

\def\ol{\overline}

\newcommand{\sgn}{\operatorname{sgn}}

\begin{document}
\markboth{Agaoglou, Fe\v{c}kan, Posp\' i\v sil,   Rothos \& H. Susanto}{TRAVELLING WAVES IN NONLINEAR MAGNETO-INDUCTIVE LATTICES}

\title{TRAVELLING WAVES IN NONLINEAR MAGNETO-INDUCTIVE LATTICES}

\author{M.\ Agaoglou}
\address{Lab of Nonlinear Mathematics and Department of Mechanical Engineering, Faculty of Engineering,\\ Aristotle University of Thessaloniki, Thessaloniki 54124, Greece}

\author{M. Fe\v{c}kan}
\address{Department of Mathematical Analysis and Numerical Mathematics, Comenius University in Bratislava,\\ Mlynsk\'a dolina, 842 48 Bratislava, Slovakia}
\address{Mathematical Institute of Slovak Academy of Sciences, \v{S}tef\'anikova 49, 814 73 Bratislava, Slovakia}

\author{M. Posp\' i\v sil}
\address{Mathematical Institute of Slovak Academy of Sciences, \v{S}tef\'anikova 49, 814 73 Bratislava, Slovakia}

\author{V.M.\ Rothos}
\address{Lab of Nonlinear Mathematics and Department of Mechanical Engineering, Faculty of Engineering,\\ Aristotle University of Thessaloniki, Thessaloniki 54124, Greece}

\author{H.\ Susanto}
\address{Department of Mathematical Sciences, University of Essex, Wivenhoe Park, Colchester CO4 3SQ,\\ United Kingdom}

\begin{abstract}
We consider a lattice equation modelling one-dimensional metamaterials formed by a discrete array of nonlinear resonators. We focus on periodic travelling waves due to the presence of a periodic force.
The existence and uniqueness results of periodic travelling waves of the system are presented. Our analytical results are found to be in good agreement with direct numerical computations.
\end{abstract}

\maketitle

\section{Introduction}\label{intro}

In this work, we consider \cite{V}
\begin{equation}\label{eq1}
	\begin{gathered}
		\frac{d^2}{dt^2}(u_n-\la u_{n-1}-\la u_{n+1})+\ga\frac{d}{dt}u_n+u_n
		+\frac{d^2}{dt^2}(u_n^2-\la u_{n-1}^2-\la u_{n+1}^2)+\ga \frac{d}{dt}u_n^2
			-h(\om t+pn)=0
	\end{gathered}
\end{equation}
where $\ga\geq 0$, $\la,\om>0$, $p\neq 0$ are parameters and $h\in C(\R,\R)$ is such that
$$h(z)=\sum_{k\in\Z}h_k \eu^{k\imath z},\qquad \sum_{k\in\Z}|h_k|<\infty.$$ The equation models the dynamics of electromagnetic waves in the so-called magneto-inductive metamaterials.

Metamaterials are artificial materials that are engineered to have properties that may not be found in nature \cite{cui09}. The (structural rather than chemical) engineering is achieved by composing periodic inhomogeneities to create desirable effective behaviour. The invention ignited a new paradigm in electromagnetism, including cloaking devices \cite{schu06} (see also \cite{litc10} for a recent comprehensive review of electromagnetic manipulation enabled by metamaterials). Magnetic metamaterials are non-magnetic materials exhibiting magnetic properties in the Terahertz and optical frequencies. They were predicted theoretically in \cite{ishikawa,zhou} and demonstrated experimentally in, e.g., \cite{enkr05,grigorenko,katsarakis,yen-moser}. When the materials composed of magnetic metamaterials, e.g., split-ring resonators (SRRs), are magnetically weakly coupled through their mutual inductances, one obtains magneto-inductive metamaterials \cite{frei04,sham02,sydo05,syms05}. Magneto-inductive metamaterials consisting of periodic arrays of  SRRs have been made in one-, two- and three dimensions \cite{marq08}. Both the dimensions of SRRs and their inter-space distance are small relative to the free space wavelength and, thus, the dynamics of electromagnetic fields in such settings are governed by the quasi magnetostatic approximation \cite{jack99}.

Model equations for the propagation of nonlinear electromagnetic waves in metamaterials can be grouped into two classes. The first approach assumes that effective metamaterials are homogeneous media with specific physical properties resulting in partial differential equations, such as coupled short-pulse equation \cite{tsit10} and higher-order nonlinear Schr\"odinger equations \cite{tsit09}. In the second class, metamaterials are modelled by arrays of coupled oscillators, i.e.\ lattice equations, such as a nonlinear Klein-Gordon equation \cite{long05} and coupled Klein-Gordon equations \cite{KLT,LET}. The governing equation \eqref{eq1} falls into the second approach with $\ga$ representing the loss coefficient, $\la$ is the coupling parameter, $h$ is an external forcing that is periodic in time and varying in the spatial direction.

In the lattice equation \eqref{eq1} the nonlinearity appears in the coupling terms and in the damping. This is due to the assumption of the nonlinearity of the capacitance of the slit-ring resonators that compose the magneto-inductive materials \cite{V}. Note that the nonlinearity is different from our previous work \cite{agao14,dibl14}, where it is in the onsite potential. The present work is to study the effect of such nonlinearity. Here, we investigate the existence of travelling periodic solutions of system \eqref{eq1} due to the periodic forcing, and the bifurcation of such solutions with small $\ga$, $\la$ and $h$.  We also study the modulational stability of the periodic solutions by computing Floquet multipliers of the linearised systems.

Note that in our governing equation \eqref{eq1} the nonlinearity in the coupling terms between the sites is akin to that in the Fermi-Pasta-Ulam lattices \cite{galla07}. However, there is a significant difference in fact that the coupling in our governing equation is also the derivative term. The bifurcation structures of periodic solutions in general FPU lattices forced by periodic drive were studied in \cite{feck13}. It is imperative to study the effect of the present nonlinear couplings to periodic solutions caused by the same periodic drive. Using Lyapunov-Schmidt reduction, we derive the asymptotic expressions of the bifurcating solutions.

The present paper is organized as follows. In Section \ref{per} we are looking for a periodic travelling wave when its amplitude is limited by a function of the magnitude of the forcing. In the section, we prove the existence of such waves rigorously. The resonance condition for the parameter values is also derived. In Section \ref{res}, we study the resonance region. Using the Lyapunov-Schmidt reduction, we show that the periodic solutions persist. Several explicit examples of the general results obtained in the previous sections are presented in Section \ref{examples}. In Section \ref{nums} we solve the governing equations numerically. Comparisons with the analytical results are presented where we obtain good agreement.

\section{Existence of small periodic solutions}\label{per}

Putting $u_n(t)=U(\om t+pn)$, $z=\om t+pn$ in \eqref{eq1}, we obtain
\begin{equation}\label{eq2}
    \begin{gathered}
		\om^2(U(z)-\la U(z-p)-\la U(z+p))''+\ga\om U'(z)+U(z)\\
		{}+\om^2(U^2(z)-\la U^2(z-p)-\la U^2(z+p))''+\ga\om(U^2(z))'-h(z)=0.
    \end{gathered}
\end{equation}
We take Banach spaces
$$
\begin{gathered}
X:=\left\{U\in C(\R,\R)\mid U(z)=\sum_{k\in \Z}c_k\eu^{k\imath z},\,
	\sum_{k\in \Z}|c_k|<\infty\right\},\\
Y:=\left\{U\in C^1(\R,\R)\mid U(z)=\sum_{k\in \Z}c_k\eu^{k\imath z},\,
	\sum_{k\in \Z\backslash\{0\}}|k||c_k|<\infty\right\},\\
Z:=\left\{U\in C^2(\R,\R)\mid U(z)=\sum_{k\in \Z}c_k\eu^{k\imath z},\,
	\sum_{k\in \Z\backslash\{0\}}k^2|c_k|<\infty\right\}
\end{gathered}
$$
with the norms
$$
\|U\|:=\sum_{k\in \Z}|c_k|,\quad
\|U\|_1:=|c_0|+\sum_{k\in \Z\backslash\{0\}}|k||c_k|,\quad
	\|U\|_2:=|c_0|+\sum_{k\in \Z\backslash\{0\}}k^2|c_k|,
$$
respectively. It is easy to verify that $Z\circlearrowleft Y\circlearrowleft X$ are compact embeddings and $\|U\|\le \|U\|_1\,\forall U\in Y$, $\|U\|_1\le \|U\|_2\,\forall U\in Z$.

The following lemma is clear.

\begin{lemma}\label{Lem1}
If $U_1,U_2\in X$ then $U_1U_2\in X$ and $\|U_1U_2\|\leq \|U_1\|\|U_2\|$.
\end{lemma}
By setting
$$
\begin{gathered}
\KK U:=\om^2\(U''(z)-\la U''(z-p)-\la U''(z+p)\)+\ga\om U'(z)+U(z),\\
\FF(U,h):=-\om^2\((U^2(z))''-\la(U^2(z-p))''-\la(U^2(z+p))''\)
	-\ga\om(U^2(z))'+h(z),\end{gathered}
$$
equation \eqref{eq2} has the form
$$
\KK U=\FF(U,h)\, .
$$

We have the next result.

\begin{lemma}\label{Lem2}
    Function 
    $\FF:Z\times X\to X$ fulfils
    \begin{equation}\label{eq3}
        \|\FF(U,h)\|\leq 2\om(2\om+4\la\om+\ga)\|U\|_2^2+\|h\|,
    \end{equation}
    \begin{equation}\label{eq4}
        \|\FF(U_1,h)-\FF(U_2,h)\|
        \leq 2\om(2\om+4\la\om+\ga)\|U_1-U_2\|_2(\|U_1\|_2+\|U_2\|_2),
    \end{equation}
    \begin{equation}\label{eq5}
        \|\FF(U,h_1)-\FF(U,h_2)\|\leq \|h_1-h_2\|
    \end{equation}
for any $U,U_1,U_2\in Z$ and $h,h_1,h_2\in X$.
\end{lemma}
\begin{proof}
Property \eqref{eq5} is obvious. Next, we use $\|U'\|=\|U\|_1\leq \|U\|_2$, $\|U''\|=\|U\|_2$ with Lemma \ref{Lem1} to derive
$$\|(U^2(z+c))''\|=2\|(U'(z+c))^2+U(z+c)U''(z+c)\|\leq 4\|U(z+c)\|_2^2$$
for $c=-p,0,p$. Moreover,
$$\|U(z\pm p)\|=\sum_{k\in\Z} |c_k||\eu^{k\imath p}|=\sum_{k\in\Z} |c_k|=\|U\|$$
for $U(z)=\sum_{k\in\Z}c_k\eu^{k\imath z}$. Similarly we derive
$$
\|(U^2(z))'\|=2\|U(z)U'(z)\|\le 2\|U\|_2^2.
$$
Hence, property \eqref{eq3} is obtained.
Analogically, we get
$$
\begin{gathered}
    \|(U_1^2(z+c)-U_2^2(z+c))''\|
    \leq \|(U_1(z+c)-U_2(z+c))''(U_1(z+c)+U_2(z+c))\|\\
    {}+2\|(U_1(z+c)-U_2(z+c))'(U_1(z+c)+U_2(z+c))'\|
    +\|(U_1(z+c)-U_2(z+c))(U_1(z+c)+U_2(z+c))''\|\\
    \leq 4\|U_1-U_2\|_2\|U_1+U_2\|_2
\end{gathered}
$$
for $c=-p,0,p$, and
$$
\begin{gathered}
    \|(U_1^2(z)-U_2^2(z))'\|\leq \|(U_1(z)-U_2(z))'(U_1(z)+U_2(z))\|\\
    {}+\|(U_1(z)-U_2(z))(U_1(z)+U_2(z))'\|\leq 2\|U_1-U_2\|_2\|U_1+U_2\|_2
\end{gathered}
$$
since $a^2-b^2=(a-b)(a+b)$ for any $a,b\in\C$. Now, \eqref{eq4} follows easily.
\end{proof}
Next, if $U\in Z$ with $U(z)=\sum_{k\in \Z}c_k\eu^{k\imath z}$ then
\begin{equation}\label{lin}
    \KK U(z)=\sum_{k\in \Z}\(1-\om^2k^2+2\la\om^2k^2\cos kp+\imath\ga\om k\) c_k\eu^{k\imath z}
\end{equation}
and so $\KK\in L(Z,X)$ with
$$
    \|\KK\|_{L(Z,X)}\leq 1+\om^2(1+2\la)+\ga\om.
$$
If
\begin{equation}\label{eq7}
    \Theta:=\inf_{k\in\Z\backslash\{0\}}
    \left\{1,\sqrt{\(\frac{1}{k^2}+\om^2(2\la\cos kp-1)\)^2+\frac{\ga^2\om^2}{k^2}}\right\}>0
\end{equation}
is a constant depending on $\ga$, $\la$, $\om$ and $p$, then we also have $\KK^{-1}\in L(X,Z)\subset L(X,Y)\subset L(X)$. So $\KK^{-1} : X\to Z$ is continuous such that
\begin{equation}\label{eq8}
    \|\KK^{-1}\|_{L(X,Z)}\leq \frac{1}{\Theta}.
\end{equation}

Now we can prove the following existence result on \eqref{eq2} when all parameters except $h$ are fixed.

\begin{theorem}\label{Th1}
Assume \eqref{eq7} along with
\begin{equation}\label{eq9}
    0<\|h\|<\frac{\Theta^2}{8\om(2\om+4\la\om+\ga)}.
\end{equation}
Then equation \eqref{eq2} has a unique solution $U(h)\in {B(\rho_h)}$ in a closed ball
$$B(\rho_h):=\{U\in Z\mid \|U\|_2\leq \rho_h\},$$
where
\begin{equation}\label{eq9.5}
    \rho_h=\frac{\Theta-\sqrt{\Theta^2-8\om(2\om+4\la\om+\ga)\|h\|}}{4\om(2\om+4\la\om+\ga)}.
\end{equation}
Moreover, $U(h)$ can be approximated by an iteration process. Finally, it holds
\begin{equation}\label{eq10}
    \|U(h_1)-U(h_2)\|_2\leq
        \frac{\|h_1-h_2\|}{\sqrt{\Theta^2-8\om(2\om+4\la\om+\ga)\max\{\|h_1\|,\|h_2\|\}}}
\end{equation}
for any $h_1,h_2\in X$ satisfying \eqref{eq9}.
\end{theorem}
\begin{proof}
We rewrite \eqref{eq2} as a parameterized fixed point problem
$$
U=\RR(U,h):=\KK^{-1}\FF(U,h)
$$
in $Z$. We already know that $\RR : Z\times X\to Z$ is continuous and by \eqref{eq3}, \eqref{eq8} such that
$$
\|\RR(U,h)\|_2\leq \frac{1}{\Theta}\(2\om(2\om+4\la\om+\ga)\|U\|_2^2+\|h\|\).
$$
Next, if there is $\rho_h>0$ such that
\begin{equation}\label{eq12}
    A(\rho_h):=\Theta\rho_h-2\om(2\om+4\la\om+\ga)\rho_h^2=\|h\|,
\end{equation}
then $\RR(\cdot,h)$ maps $B(\rho_h)$ into itself. So it remains to study \eqref{eq12}. In order to find $H>0$ -- the largest right-hand side of \eqref{eq12}
when this equation has a solution $\rho_{H}>0$, we need to solve $A(r)=H$ together with $\du A(r)=0$.
This implies
$$\rho_H=\frac{\Theta}{4\om(2\om+4\la\om+\ga)}$$
and \eqref{eq9}. So assuming \eqref{eq7}, \eqref{eq9}, we know that \eqref{eq12} has a positive solution $\rho_h<\rho_H$. We take the smallest one which clearly has the form \eqref{eq9.5}.
So $\RR(\cdot,h)$ maps $B(\rho_h)$ into itself and, moreover, by \eqref{eq4}, \eqref{eq8}
$$
\|\RR(U_1,h)-\RR(U_2,h)\|_2\leq \frac{4\om(2\om+4\la\om+\ga)\|U_1-U_2\|_2\rho_h}{\Theta}
$$
for any $U_1,U_2\in B(\rho_h)$. Hence, $\RR(\cdot,h)$ is a contraction on $B(\rho_h)$ with a contraction constant
$$
\frac{4\om(2\om+4\la\om+\ga)\rho_h}{\Theta}<\frac{4\om(2\om+4\la\om+\ga)\rho_H}{\Theta}=1.
$$
The proof of the existence and uniqueness is finished by the Banach fixed point theorem \cite{Ber}.
Next, let $h_1,h_2\in X$ satisfy \eqref{eq9}. Then $U(h_i)\in B(\rho_{h_i})\subset B(\rho_{\wt{H}})$ for $i=1,2$ and $\wt{H}:=\max\left\{\|h_1\|,\|h_2\|\right\}$.
Note $\wt{H}$ satisfies \eqref{eq9}. By \eqref{eq4}, \eqref{eq5} and \eqref{eq8}, we derive
$$
\begin{gathered}
    \|U(h_1)-U(h_2)\|_2=\|\RR(U(h_1),h_1)-\RR(U(h_2),h_2)\|_2\\
    \leq \|\RR(U(h_1),h_1)-\RR(U(h_2),h_1)\|_2+\|\RR(U(h_2),h_1)-\RR(U(h_2),h_2)\|_2\\
    \leq \frac{4\om(2\om+4\la\om+\ga)\|U(h_1)-U(h_2)\|_2\rho_{\wt{H}}+\|h_1-h_2\|}{\Theta}
\end{gathered}
$$
which implies \eqref{eq10}.
\end{proof}

\begin{remark}\label{Rem1}
In the following special cases, \eqref{eq7} holds and we can replace $\Theta$ with the corresponding $\Theta_i$ in the above considerations:

1. If $p\in 2\pi\Z$, $\la>\frac{1}{2}$, then
$$\Theta\geq \Theta_1:=\min\{1,\om^2(2\la-1)\}>0.$$

2. Let $\ga=0$. If $p\in \pi\Q$, then we can write $p=\frac{p_1}{p_2}\pi$ for some $p_1\in\Z$, $p_2\in\N$, where $p_1$, $p_2$ are relatively prime (their only common divisor is $1$). Moreover,
\begin{equation}\label{eqRem1.1}
	M_1:=\{\cos kp\mid k\in\Z\backslash\{0\}\}
	\subset \left\{\cos\frac{k\pi}{p_2}\mid k\in\{0,1,\dots,2p_2-1\}\right\}=:M_2.
\end{equation}
Indeed, for each $k\in\Z$ there exist $i,j\in\Z$ such that $0\leq i\leq 2p_2-1$ and
$kp_1=2jp_2+i$. Hence, for $\cos kp\in M_1$ we have
$$\cos kp=\cos\frac{kp_1\pi}{p_2}=\cos\frac{i\pi}{p_2}\in M_2.$$

To find a better relationship between $M_1$ and $M_2$, we need to solve $\cos \frac{\wt k_1p_1}{p_2}\pi=\cos \frac{k_2\pi}{p_2}$ for $\wt k_1\in\Z\backslash\{0\}$ and $k_2\in\{0,1,\dots,2p_2-1\}$. This is equivalent to $\frac{k_1p_1}{p_2}\pi=\frac{k_2\pi}{p_2}+2\pi l$ for $l\in \Z$ and $k_1=\pm\wt k_1$, i.e., $k_1p_1=k_2+2lp_2$. If $p_1$ is odd, then $p_1$ and $2p_2$ are relatively prime. Thus, we know \cite{BML} that there exist $\hat{k}, \hat l\in\Z$ such that $\hat{k}p_1=1+2\hat lp_2$. Consequently, $\hat{k}k_2p_1=k_2+2\hat lk_2p_2$ and $M_1=M_2$. If $p_1$ is even, then $p_1=2\bar p_1$ and $k_12\bar p_1=k_2+2lp_2$, so $k_2=2\bar k_2$ and we have $k_1\bar p_1=\bar k_2+lp_2$. Clearly $\bar p_1$ and $p_2$ are relatively prime. Thus, there exist $\bar{k}, \bar l\in\Z$ such that $\bar{k}\bar p_1=1+\bar lp_2$. Then $\bar{k}\bar k_2\bar p_1=\bar k_2+\bar l\bar k_2p_2$. Hence $M_1=\left\{\cos\frac{2k\pi}{p_2}\mid k\in\{0,1,\dots,p_2-1\}\right\}\subsetneqq M_2$.

Nevertheless, assuming
\begin{equation}\label{eqRem1.2}
	\om^2\notin\left\{\frac{1}{k^2(1-2\la\cos kp)}\right\}_{k\in\Z\backslash\{0\}}
\end{equation}
and
\begin{equation}\label{eqRem1.3}
	0\notin 2\la M_2-1
\end{equation}
for $M_2$ defined in \eqref{eqRem1.1}, we get the existence of $\de>0$ such that
$$
\sqrt{\left(\frac{1}{k^2}+\om^2(2\la\cos kp-1)\right)^2+\frac{\ga^2\om^2}{k^2}}
	\geq \om^2\left|\frac{1}{\om^2k^2}+2\la\cos kp-1\right|\geq\de>0
$$
for each $k\in\Z\backslash\{0\}$.
Thus, if $p\in\pi\Q$ and \eqref{eqRem1.2}, \eqref{eqRem1.3} are valid, then
$$\Theta\geq\min\left\{1,\inf_{k\in\Z\backslash\{0\}}\om^2\left|\frac{1}{\om^2k^2}+2\la\cos kp-1\right|\right\}=:\Theta_2>0.$$

Note that $\frac{1}{k^2(1-2\la\cos kp)}\to 0$ if $|k|\to \infty$. So if $p\in\pi\Q$ and \eqref{eqRem1.3} holds, there is at most a finite number of resonant modes $k_0\in\Z\backslash\{0\}$ determined by equation
\begin{equation}\label{eqRem1.4}
\om^2=\frac{1}{k_0^2(1-2\la\cos k_0p)}.
\end{equation}
More precisely, for given $\la$, $\om$ there is not more than $2(p_2+1)$ resonant modes $k_0$. Indeed, let $\{k_{0j}\}_{j=1}^J$ be an increasing sequence of positive resonant modes corresponding to the fixed $\la$, $\om$. Then $\{1-2\la\cos k_{0j}p\}_{j=1}^J$ has to be decreasing, since
$$k_{0j}^2(1-2\la\cos k_{0j}p)=\frac{1}{\om^2}=\text{const.}\quad \forall j=1,\dots,J,$$
i.e., $\{\cos k_{0j}p\}_{j=1}^J$ is increasing. From properties of $\cos x$ and \eqref{eqRem1.1}, it follows that the longest sequence of the form $\cos k_{0j}p$ has the values
$\{\cos\frac{k\pi}{p_2}\mid k=p_2,\dots,2p_2\}$ which has $p_2+1$ elements. Finally, to each $k_{0j}$ corresponds resonant mode $-k_{0j}$.

3. If $\ga>0$, condition \eqref{eq7} is satisfied even without the non-resonance condition \eqref{eqRem1.2}, since in each resonant mode $k_0$ we have
$$\sqrt{\left(\frac{1}{k_0^2}+\om^2(2\la\cos k_0p-1)\right)^2+\frac{\ga^2\om^2}{k_0^2}}
	=\frac{\ga\om}{|k_0|}>0$$
and there is only the finite number of resonant modes. For each non-resonant mode $k$ it holds
$$\sqrt{\left(\frac{1}{k^2}+\om^2(2\la\cos kp-1)\right)^2+\frac{\ga^2\om^2}{k^2}}
	\geq \left|\frac{1}{k^2}+\om^2(2\la\cos kp-1)\right|>0.$$
Summarizing, if $\ga>0$, $p\in\pi\Q$ and \eqref{eqRem1.3} holds, condition \eqref{eq7} is fulfilled and
$$\Theta\geq \min\left\{1,\min_{k\in M}\frac{\ga\om}{|k|},
	\inf_{k\in\Z\backslash(\{0\}\cup M)}\om^2\left|\frac{1}{\om^2k^2}+2\la\cos kp-1\right|\right\}
	=:\Theta_3>0$$
where $M$ is the set of resonant modes $k_0$.
\end{remark}

\section{Simple resonances}\label{res}

In this section, we investigate the bifurcation of a small solution of \eqref{eq2} under the assumption of a simple resonance described in Remark \ref{Rem1}.2. So, we suppose
\begin{itemize}
	\item[(R)] $p\in\pi\Q$, condition \eqref{eqRem1.3} is valid, for some $k_0\in\Z\backslash\{0\}$ equation \eqref{eqRem1.4} holds and has the only solutions $\pm k_0$,
\end{itemize}
and we consider the equation
\begin{equation}\label{eqR1}
	A(U,\ep):=\KK U-\FF(U)-\ep\GG(U)=0
\end{equation}
with
\begin{equation}\label{eqR2}
	\begin{gathered}
		\KK U:=\om^2(U''(z)-\la U''(z-p)-\la U''(z+p))+U(z)\\
		\FF(U):=-\om^2((U^2(z))''-\la(U^2(z-p))''-\la(U^2(z+p))'')\\
		\GG(U)(z):=-\ga\om(U'(z)+(U^2(z))')+h(z)
	\end{gathered}
\end{equation}
for $U$ and $\ep$ small.

We shall solve equation \eqref{eqR1} via Lyapunov-Schmidt reduction method \cite{Ch}.
Clearly, $A(0,0)=0$ and $\du_U A(0,0)=\KK$. For simplicity we denote $\ka:=|k_0|$ and $R_1^\ka:=\NN\KK$, $R_2^\ka:=\RR\KK$ the null space and the range of the operator $\KK\in L(Z,X)$, respectively. Equation \eqref{lin} with $\ga=0$ together with $\mathrm{(R)}$ yields that if $U\in Z$ with $U(z)=\sum_{k\in\Z}c_k\eu^{k\imath z}$, then
$$\KK U(z)=\sum_{k\in\Z\backslash\{\pm k_0\}}
	(1-\om^2k^2+2\la\om^2k^2\cos kp)c_k\eu^{k\imath z}.$$
Consequently,
\begin{gather*}
	R_1^\ka=\left\{c\eu^{\ka\imath z}+\bar{c}\eu^{-\ka\imath z}\mid c\in\C\right\},\qquad
	R_2^\ka=\left\{U\in X
		\mid U(z)=\sum_{k\in\Z\backslash\{\pm k_0\}}c_k\eu^{k\imath z}\right\}.
\end{gather*}
Let $Q:X\to R_2^\ka$ be the projection onto $R_2^\ka\subset X$ given by
$$QU(z)=\sum_{k\in\Z\backslash\{\pm k_0\}}c_k\eu^{k\imath z}.$$
Now, we take the decomposition $U=U_1+U_2$, $U_1\in R_1^\ka$, $U_2\in R_2^\ka\cap Z$ for $U\in Z$ and decouple \eqref{eqR1} to equations
\begin{gather}
	QA(U_1+U_2,\ep)=0,\label{eqR3}\\
	(\I -Q)A(U_1+U_2,\ep)=0.\label{eqR4}
\end{gather}
Applying the implicit function theorem to the first equation implies the existence of neighbourhoods $V_1\times V_0\subset R_1^\ka\times\R$ of the point $(0,0)$ in $R_1^\ka\times\R$, $V_2\subset R_2^\ka\cap Z$ of $0$ in $R_2^\ka\cap Z$ and a unique $C^\infty$--function $U_2:V_1\times V_0\to V_2$ such that $QA(U_1+\wt{U}_2,\ep)=0$ for $(U_1,\ep)\in V_1\times V_0$ and $\wt{U}_2\in V_2$ if and only if $\wt{U}_2=U_2(U_1,\ep)$. Moreover, $U_2(0,0)=0$.

Thus, equation \eqref{eqR4} has the form
\begin{equation}\label{eqR5}
	(\I-Q)A(U_1+U_2(U_1,\ep),\ep)=0
\end{equation}
for $U_1\in V_1$ and $\ep\in V_0$. Differentiating equation \eqref{eqR3} we obtain
\begin{equation}\label{eqR5b}
\du_{U_1}U_2(0,0)=0,\qquad \du_\ep U_2(0,0)=-\KK^{-1}Q\du_\ep A(0,0)=\KK^{-1}Q\GG(0).
\end{equation}
Hence
\begin{equation}\label{eqU2exp}
	U_2(U_1,\ep)=\KK^{-1}Q\GG(0)\ep+O(\|(U_1,\ep)\|^2).
\end{equation}
Expanding $\FF$ and $\GG$ of \eqref{eqR2} into Taylor series using the last expansion of $U_2(U_1,\ep)$ yields
\begin{gather*}
	\FF(U_1+U_2(U_1,\ep))=\frac{1}{2}\du^2 \FF(0)(U_1+\du_\ep U_2(0,0)\ep)^2+O(\|(U_1,\ep)\|^3),\\
	\GG(U_1+U_2(U_1,\ep))=\GG(0)+\du \GG(0)(U_1+\du_\ep U_2(0,0)\ep)+O(\|(U_1,\ep)\|^2).
\end{gather*}
Note that $\du^2 \FF(0)V^2=2\FF(V)$ and $\ep O(\|(U_1,\ep)\|^2)=O(\|(U_1,\ep)\|^3)$. Thus equation \eqref{eqR5} is equivalent to
\begin{equation*}
	\begin{gathered}
		(\I-Q)(\FF(U_1+\du_\ep U_2(0,0)\ep)+\ep \GG(0)
		+\ep\du \GG(0)(U_1+\du_\ep U_2(0,0)\ep)+O(\|(U_1,\ep)\|^3))=0
	\end{gathered}
\end{equation*}
or, using $\du_\ep U_2(0,0)$ and $\FF$ of \eqref{eqR2},
\begin{equation}\label{eqR5.1}
	\begin{gathered}
		(\I-Q)(\om^2(((U_1(z)+\KK^{-1}Q\GG(0)(z)\ep)^2)''-\la((U_1(z-p)+\KK^{-1}Q\GG(0)(z-p)\ep)^2)''\\
	{}-\la((U_1(z+p)+\KK^{-1}Q\GG(0)(z+p)\ep)^2)'')-\ep \GG(0)(z)
		-\ep\du \GG(0)(z)(U_1+\KK^{-1}Q\GG(0)\ep)+O(\|(U_1,\ep)\|^3))=0.
	\end{gathered}
\end{equation}
Note that
\begin{equation*}
	\begin{gathered}
		((U_1(z)+\KK^{-1}Q\GG(0)(z)\ep)^2)''
		=2(U_1'(z)+(\KK^{-1}Q\GG(0)(z))'\ep)^2\\
		+2(U_1(z)+\KK^{-1}Q\GG(0)(z)\ep)(U_1''(z)+(\KK^{-1}Q\GG(0)(z))''\ep)
	\end{gathered}
\end{equation*}
for any $z\in\R$. Let $U_1(z)=c\eu^{k_0\imath z}+\bar{c}\eu^{-k_0\imath z}\in R_1^\ka$ and
$h(z)=\sum_{k\in\Z}h_k\eu^{k\imath z}\in X$. Then
\begin{equation}\label{eqR5c}
\begin{gathered}
	U_1'(z)=\tilde{c}\eu^{k_0\imath z}+\ol{\tilde{c}}\eu^{-k_0\imath z}\in R_1^\ka,\quad
	\KK^{-1}Q\GG(0)(z)=\sum_{k\in\Z\backslash\{\pm k_0\}}d_k\eu^{k\imath z}\in R_2^\ka,\quad
	(\KK^{-1}Q\GG(0)(z))'=\sum_{k\in\Z\backslash\{0,\pm k_0\}}\tilde{d}_k\eu^{k\imath z}\in R_2^\ka
\end{gathered}
\end{equation}
where
$$\tilde{c}=\imath k_0c,\qquad d_k=\frac{h_k}{1-\om^2k^2+2\la\om^2k^2\cos kp},\qquad
	\tilde{d}_k=\imath kd_k.$$
Clearly, $d_{-k}=\ol{d_k}$, $\tilde{d}_{-k}=\ol{\tilde{d}_k}$ for each $k\in\Z\backslash\{0,\pm k_0\}$.
It can be easily shown that
\begin{equation}\label{eqR6}
	\begin{gathered}
		(\I-Q)(U_1'(z))^2=0=(\I-Q)(U_1(z)U_1''(z)),\\
		(\I-Q)(U_1'(z)(\KK^{-1}Q\GG(0)(z))')=\tilde{c}\ol{\tilde{d}_{2k_0}}\eu^{-k_0\imath z}
			+\ol{\tilde{c}}\tilde{d}_{2k_0}\eu^{k_0\imath z}.
	\end{gathered}
\end{equation}
Next, setting $\tilde{d}_{\pm k_0}=\tilde{d}_0=0$,
$$\sum_{k\in\Z\backslash\{0,\pm k_0\}}\tilde{d}_k\eu^{k\imath z}
	\sum_{j\in\Z\backslash\{0,\pm k_0\}}\tilde{d}_j\eu^{j\imath z}
	=\sum_{l\in\Z}\sum_{k+j=l}\tilde{d}_k\tilde{d}_j\eu^{l\imath z}.$$
Hence the coefficient by $\eu^{k_0\imath z}$ in $((\KK^{-1}Q\GG(0)(z))')^2$ is
$$b:=\sum_{k+j=k_0}\tilde{d}_k\tilde{d}_j=\sum_{k\in\Z}\tilde{d}_k\tilde{d}_{k_0-k}
	=\sum_{k\in\Z}\tilde{d}_k\ol{\tilde{d}_{k-k_0}}
	=\sum_{k\in\Z\backslash\{0,\pm k_0,2k_0\}}\tilde{d}_k\ol{\tilde{d}_{k-k_0}}.$$
Therefore
\begin{equation}\label{eqR7}
	(\I-Q)(((\KK^{-1}Q\GG(0)(z))')^2)=b\eu^{k_0\imath z}+\ol{b}\eu^{-k_0\imath z}.
\end{equation}
Analogically, denoting
$$\tilde{\tilde{d}}_k:=-k^2d_k,\; k\in\Z\backslash\{0,\pm k_0\},
\qquad \tilde{\tilde{c}}:=-k_0^2c,\qquad
\tilde{b}:=\sum_{k\in\Z\backslash\{0,\pm k_0,2k_0\}}
	d_k\ol{\tilde{\tilde{d}}_{k-k_0}}$$
we have
\begin{gather*}
	U_1''(z)=\tilde{\tilde{c}}\eu^{k_0\imath z}
		+\ol{\tilde{\tilde{c}}}\eu^{-k_0\imath z}\in R_1^\ka,\qquad
	(\KK^{-1}Q\GG(0)(z))''=\sum_{k\in\Z\backslash\{0,\pm k_0\}}
		\tilde{\tilde{d}}_k\eu^{k\imath z}\in R_2^\ka,
\end{gather*}
and the next identities can be proved
\begin{equation}\label{eqR8}
	\begin{gathered}
		(\I-Q)(U_1(z)(\KK^{-1}Q\GG(0)(z))'')=\ol{c}\tilde{\tilde{d}}_{2k_0}\eu^{k_0\imath z}
			+c\ol{\tilde{\tilde{d}}_{2k_0}}\eu^{-k_0\imath z},\\
		(\I-Q)(U_1''(z)\KK^{-1}Q\GG(0)(z))=\(\ol{\tilde{\tilde{c}}}d_{2k_0}+\tilde{\tilde{c}}d_0\)
			\eu^{k_0\imath z}+\(\tilde{\tilde{c}}\ol{d_{2k_0}}+\ol{\tilde{\tilde{c}}}d_0\)
			\eu^{-k_0\imath z},\\
		(\I-Q)(\KK^{-1}Q\GG(0)(z)(\KK^{-1}Q\GG(0)(z))'')
			=\tilde{b}\eu^{k_0\imath z}+\ol{\tilde{b}}\eu^{-k_0\imath z}.
	\end{gathered}
\end{equation}
Moreover, by definition of $\GG$ in \eqref{eqR2},
\begin{equation}\label{eqR9}
	\begin{gathered}
		(\I-Q)\GG(0)(z)=h_{k_0}\eu^{k_0\imath z}+\ol{h_{k_0}}\eu^{-k_0\imath z},\\
		(\I-Q)\du \GG(0)(U_1+\KK^{-1}Q\GG(0)(z)\ep)=-\ga\om(\tilde{c}\eu^{k_0\imath z}
			+\ol{\tilde{c}}\eu^{-k_0\imath z}).
	\end{gathered}
\end{equation}
Now, we put \eqref{eqR6}, \eqref{eqR7}, \eqref{eqR8}, \eqref{eqR9} in \eqref{eqR5.1} and collect coefficients by $\eu^{k_0\imath z}$ and $\eu^{-k_0\imath z}$ to obtain
\begin{equation}\label{eqR10}
	H\eu^{k_0\imath z}+\ol{H}\eu^{-k_0\imath z}+O(\|(U_1,\ep)\|^3)=0
\end{equation}
where
\begin{gather*}
	H=2\om^2\(2\ol{\tilde{c}}\tilde{d}_{2k_0}+\ol{c}\tilde{\tilde{d}}_{2k_0}
		+\ol{\tilde{\tilde{c}}}d_{2k_0}+\tilde{\tilde{c}}d_0\)
		(1-\la\eu^{-k_0\imath p}-\la\eu^{k_0\imath p})\ep
	-(h_{k_0}-\ga\om\tilde{c})\ep
		+2\om^2(b+\tilde{b})(1-\la\eu^{-k_0\imath p}-\la\eu^{k_0\imath p})\ep^2.
\end{gather*}
Equation \eqref{eqR10} is equivalent to
\begin{equation}\label{eqR10c}
H+O(\|(c,\ol{c},\ep)\|^3)=0.
\end{equation}
Using resonance condition \eqref{eqRem1.4}, we get
$$(1-\la\eu^{-k_0\imath p}-\la\eu^{k_0\imath p})=1-2\la\cos k_0p=\frac{1}{\om^2k_0^2}.$$
Dropping tildes yields
\begin{gather*}
	2\ol{\tilde{c}}\tilde{d}_{2k_0}+\ol{c}\tilde{\tilde{d}}_{2k_0}
		+\ol{\tilde{\tilde{c}}}d_{2k_0}+\tilde{\tilde{c}}d_0
		=-k_0^2\ol{c}d_{2k_0}-k_0^2cd_0,\\
	b+\tilde{b}=\sum_{k\in\Z\backslash\{0,\pm k_0,2k_0\}} k_0(k-k_0)d_k\ol{d_{k-k_0}}.
\end{gather*}
Thus we have to solve
\begin{gather*}
	-\(2(\ol{c}d_{2k_0}+cd_0)+h_{k_0}-\ga\om\tilde{c}\)\ep
	+\frac{2\ep^2}{k_0^2}\sum_{k\in\Z\backslash\{0,\pm k_0,2k_0\}} k_0(k-k_0)d_k\ol{d_{k-k_0}}
	+O(\|(c,\ol{c},\ep)\|^3)=0.
\end{gather*}
Since $c,h_k\in\C$, we have
$$c=x+\imath y,\qquad h_k=\mu_k+\imath\nu_k,\qquad d_k=\frac{\mu_k+\imath\nu_k}{\de_k}$$
where $x,y,\mu_k,\nu_k\in\R$ and $\de_k=1-\om^2k^2+2\la\om^2k^2\cos kp\in\R$. Obviously, $\de_0=1$, $h_0=d_0=\mu_0\in\R$ and $\mu_k=\mu_{-k}$, $\nu_k=-\nu_{-k}$ $\forall k\in\N$. So, we get
\begin{gather*}
	-\(\frac{2(x-\imath y)(\mu_{2k_0}+\imath\nu_{2k_0})}{\de_{2k_0}}
		+2(x+\imath y)\mu_0+\mu_{k_0}+\imath\nu_{k_0}
		-\ga\om\imath k_0(x+\imath y)\)\ep\\
	+\frac{2\ep^2}{k_0^2}\sum_{k\in\Z\backslash\{0,\pm k_0,2k_0\}}
		\frac{k_0(k-k_0)(\mu_k+\imath\nu_k)(\mu_{k-k_0}-\imath\nu_{k-k_0})}{\de_k\de_{k-k_0}}
		+O(\|(x,y,\ep)\|^3)=0
\end{gather*}
or, separating the real and imaginary parts,
\begin{align*}
	-\(\frac{2(x\mu_{2k_0}+y\nu_{2k_0})}{\de_{2k_0}}+2x\mu_0+\mu_{k_0}
		+\ga\om k_0y\)\ep
	+\frac{2\ep^2}{k_0^2}\sum_{k\in\Z\backslash\{0,\pm k_0,2k_0\}}
		\frac{k_0(k-k_0)(\mu_k\mu_{k-k_0}+\nu_k\nu_{k-k_0})}{\de_k\de_{k-k_0}}
		+O(\|(x,y,\ep)\|^3) &= 0,\\
	-\(\frac{2(x\nu_{2k_0}-y\mu_{2k_0})}{\de_{2k_0}}+2y\mu_0+\nu_{k_0}
		-\ga\om k_0x\)\ep
	+\frac{2\ep^2}{k_0^2}\sum_{k\in\Z\backslash\{0,\pm k_0,2k_0\}}
		\frac{k_0(k-k_0)(\mu_{k-k_0}\nu_k-\mu_k\nu_{k-k_0})}{\de_k\de_{k-k_0}}
		+O(\|(x,y,\ep)\|^3) &= 0.
\end{align*}
Multiplying both equations $(-1)$ and collecting terms by $x$ and $y$ we obtain
\begin{equation}\label{bif1}
\begin{aligned}
	2\(\frac{\mu_{2k_0}}{\de_{2k_0}}+\mu_0\)\ep x
		+\(\frac{2\nu_{2k_0}}{\de_{2k_0}}+\ga\om k_0\)\ep y+\mu_{k_0}\ep
	-\frac{2\ep^2}{k_0^2}\sum_{k\in\Z\backslash\{0,\pm k_0,2k_0\}}
		\frac{k_0(k-k_0)(\mu_k\mu_{k-k_0}+\nu_k\nu_{k-k_0})}{\de_k\de_{k-k_0}}
		+O(\|(x,y,\ep)\|^3)=0,\\
	\(\frac{2\nu_{2k_0}}{\de_{2k_0}}-\ga\om k_0\)\ep x
		+2\(-\frac{\mu_{2k_0}}{\de_{2k_0}}+\mu_0\)\ep y+\nu_{k_0}\ep
	-\frac{2\ep^2}{k_0^2}\sum_{k\in\Z\backslash\{0,\pm k_0,2k_0\}}
		\frac{k_0(k-k_0)(\mu_{k-k_0}\nu_k-\mu_k\nu_{k-k_0})}{\de_k\de_{k-k_0}}
		+O(\|(x,y,\ep)\|^3)=0.
\end{aligned}
\end{equation}
Now we intend to derive the cubic $O(\|(U_1,0)\|^3)$ of \eqref{eqR10} as follows. We denote by $B:=-\frac{1}{2}\du^2\FF(0)$ and $u_2(U_1):=U_2(U_1,0)$. Then \eqref{eqR3} and \eqref{eqR4} are equivalent to
\begin{gather}
	Q(\KK u_2(U_1)+B(U_1+u_2(U_1))^2)=0,\label{eqR3b}\\
	(\I -Q)B(U_1+u_2(U_1))^2=0.\label{eqR4b}
\end{gather}
First note that $u_2(U_1)=O(\|U_1\|^2)$ by \eqref{eqU2exp}, $(\I -Q)BU_1^2=0$ and then \eqref{eqR4b} has the form
\begin{gather*}
0=(\I -Q)B(U_1+u_2(U_1))^2
=(\I -Q)BU_1^2+2(I-Q)B(U_1,u_2(U_1))+(I-Q)B(u_2(U_1))^2\\
=(I-Q)B(U_1,\du_{U_1}^2u_2(0)U_1^2)+O(\|U_1\|^4),
\end{gather*}
hence $O(\|(U_1,0)\|^3)=(I-Q)B(U_1,\du_{U_1}^2u_2(0)U_1^2)$. To derive $\du_{U_1}^2u_2(0)$, we twice differentiate \eqref{eqR3b} to obtain
$$
	\KK \du_{U_1}u_2(U_1)V_1+2QB(U_1+u_2(U_1),V_1+\du_{U_1}u_2(U_1)V_1)=0
$$
and then
$$
    \KK \du_{U_1}^2u_2(0)(V_1,V_2)+2QB(V_1,V_2)=0,
$$
hence
$$
\du_{U_1}^2u_2(0)U_1^2=-2\KK^{-1}QBU_1^2=2\KK^{-1}Q\FF(U_1),
$$
consequently, we arrive at
\begin{equation}\label{eqp1}
	\begin{gathered}
		O(\|(U_1,0)\|^3)=(I-Q)B(U_1,\du_{U_1}^2u_2(0)U_1^2)
		=-(I-Q)\du^2\FF(0)(U_1,\KK^{-1}Q\FF(U_1)).
	\end{gathered}
\end{equation}
Let $U_1(z)=c\eu^{k_0\imath z}+\bar{c}\eu^{-k_0\imath z}\in R_1^\ka$. Then
\begin{gather*}
\KK^{-1}Q\FF(U_1)=\KK^{-1}Q\(4\om^2k_0^2(1-2\la\cos2k_0p)\(c^2\eu^{2k_0\imath z}+\bar{c}^2\eu^{-2k_0\imath z}\)\)
=\frac{1-\de_{2k_0}}{\de_{2k_0}}c^2\eu^{2k_0\imath z}
	+\frac{1-\de_{2k_0}}{\de_{2k_0}}\bar{c}^2\eu^{-2k_0\imath z}.
\end{gather*}
Next we derive
\begin{gather*}
	\du^2\FF(0)(V,W)
	=-2\om^2((V(z)W(z))''-\la(V(z-p)W(z-p))''-\la(V(z+p)W(z+p))'')\\
	=-2\om^2\Big(V(z)W''(z)+V''(z)W(z)+2V'(z)W'(z)\\
	{}-\la\big(V(z-p)W''(z-p)+V''(z-p)W(z-p)+2V'(z)W'(z)\big)\\
	{}-\la\big(V(z+p)W''(z+p)+V''(z+p)W(z+p)+2V'(z)W'(z)\big)\Big).
\end{gather*}
Denoting $\mu=\frac{1-\de_{2k_0}}{\de_{2k_0}}$ we get
\begin{gather*}
	U_1(z)[\KK^{-1}Q\FF(U_1)(z)]''=-4k_0^2 U_1(z)[\KK^{-1}Q\FF(U_1)(z)]
	=-4k_0^2\Big(\mu c^3\eu^{3k_0\imath z}+\mu\bar{c}^3\eu^{-3k_0\imath z}
		+\mu c^2\bar{c}\eu^{k_0\imath z}+\mu c\bar{c}^2\eu^{-k_0\imath z}\Big),\\
	U_1''(z)[\KK^{-1}Q\FF(U_1)(z)]=-k_0^2 U_1(z)[\KK^{-1}Q\FF(U_1)(z)]
	=-k_0^2\Big(\mu c^3\eu^{3k_0\imath z}+\mu\bar{c}^3\eu^{-3k_0\imath z}
		+\mu c^2\bar{c}\eu^{k_0\imath z}+\mu c\bar{c}^2\eu^{-k_0\imath z}\Big),\\
	U_1'(z)[\KK^{-1}Q\FF(U_1)(z)]'
		=\Big(k_0\imath c\eu^{k_0\imath z}-k_0\imath\bar{c}\eu^{-k_0\imath z}\Big)
		\Big(2k_0\imath\mu c^2\eu^{2k_0\imath z}-2k_0\imath\mu\bar{c}^2\eu^{-2k_0\imath z}\Big)\\
	=-2k_0^2\mu c^3\eu^{3k_0\imath z}-2k_0^2\mu\bar{c}^3\eu^{-3k_0\imath z}
		+2k_0^2\mu c^2\bar{c}\eu^{k_0\imath z}+2k_0^2\mu c\bar{c}^2\eu^{-k_0\imath z}.
\end{gather*}
So
$$(I-Q)\Big(U_1(z)[\KK^{-1}Q\FF(U_1)(z)]''+2U_1'(z)[\KK^{-1}Q\FF(U_1)(z)]'\Big)=0.$$
Hence
\begin{gather*}
	-(I-Q)\du^2\FF(0)(U_1,\KK^{-1}Q\FF(U_1))
	=2\om^2\Big(\mu c^2\bar{c}\eu^{k_0\imath z}(-k_0^2)+\mu c\bar{c}^2\eu^{-k_0\imath z}(-k_0^2)\\
	{}-\la\mu c^2\bar{c}\eu^{k_0\imath (z-p)}(-k_0^2)
		-\la\mu c\bar{c}^2\eu^{-k_0\imath (z-p)}(-k_0^2)
	-\la\mu c^2\bar{c}\eu^{k_0\imath (z+p)}(-k_0^2)
		-\la\mu c\bar{c}^2\eu^{-k_0\imath (z+p)}(-k_0^2)\Big)\\
	=-2k_0^2\om^2\mu|c|^2\Big(U_1(z)-\la U_1(z-p)-\la U_1(z+p)\Big)
	=-2k_0^2\om^2\mu|c|^2(1-\la\eu^{k_0\imath p}-\la \eu^{-k_0\imath p})U_1(z)\\
	=-2\mu|c|^2U_1(z)
		=-\frac{2(1-\de_{2k_0})}{\de_{2k_0}}|c|^2 \(c\eu^{k_0\imath z}+\bar{c}\eu^{-k_0\imath z}\).
\end{gather*}
Thus, by \eqref{eqp1} and the last identity, in \eqref{eqR10c} we have
$$
O(\|(c,\ol{c},0)\|^3)=-\frac{2c|c|^2(1-\de_{2k_0})}{\de_{2k_0}}.
$$
Summarizing, the bifurcation equation \eqref{bif1} has a form
\begin{equation}\label{bif2}
\begin{gathered}
	2\(\frac{\mu_{2k_0}}{\de_{2k_0}}+\mu_0\)\ep x
		+\(\frac{2\nu_{2k_0}}{\de_{2k_0}}+\ga\om k_0\)\ep y+\mu_{k_0}\ep+\frac{2(1-\de_{2k_0})x(x^2+y^2)}{\de_{2k_0}}\\
	-\frac{2\ep^2}{k_0^2}\sum_{k\in\Z\backslash\{0,\pm k_0,2k_0\}}
		\frac{k_0(k-k_0)(\mu_k\mu_{k-k_0}+\nu_k\nu_{k-k_0})}{\de_k\de_{k-k_0}}
		+\ep O(\|(x,y,\ep)\|^2)=0,\\
	\(\frac{2\nu_{2k_0}}{\de_{2k_0}}-\ga\om k_0\)\ep x
		+2\(-\frac{\mu_{2k_0}}{\de_{2k_0}}+\mu_0\)\ep y+\nu_{k_0}\ep+\frac{2(1-\de_{2k_0})y(x^2+y^2)}{\de_{2k_0}}\\
	-\frac{2\ep^2}{k_0^2}\sum_{k\in\Z\backslash\{0,\pm k_0,2k_0\}}
		\frac{k_0(k-k_0)(\mu_{k-k_0}\nu_k-\mu_k\nu_{k-k_0})}{\de_k\de_{k-k_0}}
		+\ep O(\|(x,y,\ep)\|^2)=0.
\end{gathered}
\end{equation}
Now, we scale \eqref{bif2}
\begin{equation}\label{eqScala}
	x\longleftrightarrow\ep x,\qquad y\longleftrightarrow\ep y,
		\qquad \ep\longleftrightarrow\ep^3,\qquad \ga\longleftrightarrow\ga
\end{equation}
and divide both equations by $\ep^3$. This leads to
\begin{equation}\label{eqR11a}
	\begin{aligned}
		\mu_{k_0}+\frac{2(1-\de_{2k_0})x(x^2+y^2)}{\de_{2k_0}}+O(\ep) &= 0\\
		\nu_{k_0}+\frac{2(1-\de_{2k_0})y(x^2+y^2)}{\de_{2k_0}}+O(\ep) &= 0
	\end{aligned}
\end{equation}
and the following result holds.

\begin{theorem}\label{Th2}
Assume $\mathrm{(R)}$ and $(\mu_{k_0},\nu_{k_0})\ne(0,0)$. Then equation \eqref{eqR1} has a solution $U\in Z$ close to $0$ for any $\ep\neq 0$ sufficiently small.
\end{theorem}

\begin{proof}
Let us denote
\begin{equation*}
	\wt{H}_1(x,y,\ep) := \(\mu_{k_0}+\frac{2(1-\de_{2k_0})x(x^2+y^2)}{\de_{2k_0}}+O(\ep),
		\nu_{k_0}+\frac{2(1-\de_{2k_0})y(x^2+y^2)}{\de_{2k_0}}+O(\ep)\).
\end{equation*}
Implicit function theorem applied to equation \eqref{eqR11a}, i.e.
$\wt{H}_1(x,y,\ep)=0$ gives the existence of neighbourhoods $W_0\subset V_0\subset\R$ of the point $0$ in $\R$, $W\subset\R^2$ of
$-\sqrt[3]{\frac{\de_{2k_0}}{2(1-\de_{2k_0})(\mu_{k_0}^2+\nu_{k_0}^2)}}(\mu_{k_0},\nu_{k_0})$ in $\R^2$ and a unique pair of continuous functions $x(\ep)$, $y(\ep)$ defined in $W_0$ such that
$\wt{H}_1(\tilde{x},\tilde{y},\ep)=0$ for $\ep\in W_0$, $(\tilde{x},\tilde{y})\in W$ if and only if $\tilde{x}=x(\ep)$ and $\tilde{y}=y(\ep)$. Moreover,
$(x(0),y(0))=-\sqrt[3]{\frac{\de_{2k_0}}{2(1-\de_{2k_0})(\mu_{k_0}^2+\nu_{k_0}^2)}}(\mu_{k_0},\nu_{k_0})$.
After scaling backwards in \eqref{eqScala}, we have the desired solution of equation \eqref{eqR1}.
\end{proof}

Now, assume $\ep h$ instead of $h$ in \eqref{eqR1}. Then we scale \eqref{bif2}
\begin{equation}\label{eqScalb}
	x\longleftrightarrow\ep x,\qquad y\longleftrightarrow\ep y,
		\qquad h_k\longleftrightarrow\ep h_k,\qquad \ga\longleftrightarrow\ga
\end{equation}
and divide both equations by $\ep^2$. This leads to
\begin{equation}\label{eqR11b}
	\begin{aligned}
		\ga\om k_0y+\mu_{k_0}+O(\ep) &= 0\\
		-\ga\om k_0x+\nu_{k_0}+O(\ep) &= 0
	\end{aligned}
\end{equation}
and the following result.

\begin{theorem}\label{Th3}
Assume $\mathrm{(R)}$, $\ga\neq 0$ and $\ep h$ instead of $h$ in \eqref{eqR1} with $(\mu_{k_0},\nu_{k_0})\ne(0,0)$. Then equation \eqref{eqR1} has a solution $U\in Z$ close to $0$ for any $\ep\neq 0$ sufficiently small.
\end{theorem}

\begin{proof}
Let us denote
\begin{gather*}
	\wt{H}_2(x,y,\ep):=\(\ga\om k_0y+\mu_{k_0}+O(\ep),-\ga\om k_0x+\nu_{k_0}+O(\ep)\).
\end{gather*}
Implicit function theorem applied to equation \eqref{eqR11b}, i.e.
$\wt{H}_2(x,y,\ep)=0$ gives the existence of neighbourhoods $W_0\subset V_0\subset\R$ of the point $0$ in $\R$, $W\subset\R^2$ of
$(\frac{\nu_{k_0}}{\ga\om k_0},-\frac{\mu_{k_0}}{\ga\om k_0})$ in $\R^2$ and a unique pair of continuous functions $x(\ep)$, $y(\ep)$ defined in $W_0$ such that
$\wt{H}_2(\tilde{x},\tilde{y},\ep)=0$ for $\ep\in W_0$, $(\tilde{x},\tilde{y})\in W$ if and only if $\tilde{x}=x(\ep)$ and $\tilde{y}=y(\ep)$. Moreover,
$x(0)=\frac{\nu_{k_0}}{\ga\om k_0}$, $y(0)=-\frac{\mu_{k_0}}{\ga\om k_0}$.
After scaling backwards in \eqref{eqScalb}, we have the desired solution of equation \eqref{eqR1}.
\end{proof}

Finally, we deal with a gap between Theorems \ref{Th2} and \ref{Th3}, namely that $h$ is not scaled and $\mu_{k_0}=\nu_{k_0}=0$. Now, we scale \eqref{bif2}
\begin{equation}\label{eqScalc}
	x\longleftrightarrow\ep x,\qquad y\longleftrightarrow\ep y,
		\qquad \ep\longleftrightarrow\pm\ep^2,\qquad \ga\longleftrightarrow\ga
\end{equation}
and divide both equations by $\ep^3$. This leads to
\begin{equation}\label{eqR11c}
	\begin{aligned}
	\pm2\(\frac{\mu_{2k_0}}{\de_{2k_0}}+\mu_0\)x
		\pm\(\frac{2\nu_{2k_0}}{\de_{2k_0}}+\ga\om k_0\)y
		+\frac{2(1-\de_{2k_0})x(x^2+y^2)}{\de_{2k_0}} +O(\ep) &= 0,\\
	\pm\(\frac{2\nu_{2k_0}}{\de_{2k_0}}-\ga\om k_0\)x
		\pm2\(-\frac{\mu_{2k_0}}{\de_{2k_0}}+\mu_0\) y
		+\frac{2(1-\de_{2k_0})y(x^2+y^2)}{\de_{2k_0}}+O(\ep) &= 0
    \end{aligned}
\end{equation}
and the following result.

\begin{theorem}\label{Th4}
Assume $\mathrm{(R)}$ and $(\mu_{k_0},\nu_{k_0})=(0,0)$. Then equation \eqref{eqR1} has a solution $U\in Z$ close to $0$ for any $\ep\neq 0$ sufficiently small provided
\begin{equation}\label{con1}
\de_{2k_0}^2(4\mu_0^2+\ga^2\om^2k_0^2)-4(\mu_{2k_0}^2+\nu_{2k_0}^2)\ne0.
\end{equation}
The order of this solution is $O(\ep)$. If, in addition, it holds
\begin{equation}\label{con2}
	\de_{2k_0}^2(4\mu_0^2+\ga^2\om^2k_0^2)-4(\mu_{2k_0}^2+\nu_{2k_0}^2)<0,
\end{equation}
then \eqref{eqR1} has another solution. This one is of order $O(\sqrt{|\ep|})$.
\end{theorem}

\begin{proof}
Let us denote
\begin{align*}
	\wt{H}_3(x,y,\ep) &:=
   \Bigg(\pm2\(\frac{\mu_{2k_0}}{\de_{2k_0}}+\mu_0\)x
		\pm\(\frac{2\nu_{2k_0}}{\de_{2k_0}}+\ga\om k_0\)y
		+\frac{2(1-\de_{2k_0})x(x^2+y^2)}{\de_{2k_0}} +O(\ep),\\
	&\qquad \pm\(\frac{2\nu_{2k_0}}{\de_{2k_0}}-\ga\om k_0\)x
		\pm2\(-\frac{\mu_{2k_0}}{\de_{2k_0}}+\mu_0\)y
		+\frac{2(1-\de_{2k_0})y(x^2+y^2)}{\de_{2k_0}}+O(\ep)\Bigg).
\end{align*}
First we note that $\wt{H}_3(x,y,\ep)=A(x,y)+B(x,y)+O(\ep)$ where $A$ is linear and $B$ is cubic. Since \eqref{con1} is equivalent to $\det A\ne0$, implicit function theorem applied to equation \eqref{eqR11c}, i.e.
$\wt{H}_3(x,y,\ep)=0$ gives the existence of neighbourhoods $W_0\subset V_0\subset\R$ of the point $0$ in $\R$, $W\subset\R^2$ of
$(0,0)$ in $\R^2$ and a unique pair of continuous functions $x(\ep)$, $y(\ep)$ defined in $W_0$ such that
$\wt{H}_3(\tilde{x},\tilde{y},\ep)=0$ for $\ep\in W_0$, $(\tilde{x},\tilde{y})\in W$ if and only if $\tilde{x}=x(\ep)$ and $\tilde{y}=y(\ep)$. Moreover, $(x(0),y(0))=(0,0)$, so $x(\ep)=O(\ep)$, $y(\ep)=O(\ep)$, which by scaling \eqref{eqScalc} gives $O(\ep)$ solution for \eqref{eqR1}. On the other hand, if \eqref{con2} holds then the Brouwer degree $\deg(\wt H_3,B_1,(0,0))=\sgn\det A=-1$ while $\deg(\wt H_3,B_2,(0,0))=1$ for small and large balls $B_1$, $B_2$ in $\R^2$ centered at $(0,0)$. This gives another solution of \eqref{eqR11c} of order $O(1)$ placed in the set $B_2\setminus B_1$. After scaling backwards in \eqref{eqScalc}, we have the desired $O(\sqrt{|\ep|})$ solution of \eqref{eqR1}.
\end{proof}

\begin{remark} The above $O(\ep)$ solution for \eqref{eqR1} can be derived by putting Taylor expansions $x(\ep)=x'(0)\ep+O(\ep^2)$, $y(\ep)=y'(0)\ep+O(\ep^2)$ into \eqref{bif2} and comparing $\ep^2$-terms. So we derive
\begin{equation}\label{der}
\begin{aligned}
	2\(\frac{\mu_{2k_0}}{\de_{2k_0}}+\mu_0\)x'(0)
		+\(\frac{2\nu_{2k_0}}{\de_{2k_0}}+\ga\om k_0\)y'(0)
	&= \frac{2}{k_0^2}\sum_{k\in\Z\backslash\{0,\pm k_0,2k_0\}}
		\frac{k_0(k-k_0)(\mu_k\mu_{k-k_0}+\nu_k\nu_{k-k_0})}{\de_k\de_{k-k_0}},\\
	\(\frac{2\nu_{2k_0}}{\de_{2k_0}}-\ga\om k_0\)x'(0)
		+2\(-\frac{\mu_{2k_0}}{\de_{2k_0}}+\mu_0\)y'(0)
	&= \frac{2}{k_0^2}\sum_{k\in\Z\backslash\{0,\pm k_0,2k_0\}}
		\frac{k_0(k-k_0)(\mu_{k-k_0}\nu_k-\mu_k\nu_{k-k_0})}{\de_k\de_{k-k_0}},
\end{aligned}
\end{equation}
which determines the first order approximation of this solution. Hence by \eqref{eqU2exp} and \eqref{eqR5c}, we get
\begin{equation}\label{sol1}
U(z)=U_1(z)+U_2(U_1,\ep)=\ep\(c'(0)\eu^{k_0\imath z}+\overline{c'(0)}\eu^{-k_0\imath z}+\sum_{k\in\Z\backslash\{\pm k_0\}}d_k\eu^{k\imath z}\)+O(\ep^2)
\end{equation}
for $c'(0)=x'(0)+y'(0)\imath$.
\end{remark}

\begin{remark}\label{Rem2}\

1. Now we specify the $O(\sqrt{|\ep|})$ solution (solutions) of Theorem \ref{Th4}. Nontrivial roots of equation
\begin{equation}\label{eqTh4_1}
	\wt{H}_3(x,y,0)=0
\end{equation}
can be explicitly calculated as follows. Denoting
\begin{equation}\label{abcd}
	\begin{gathered}
		a:=\pm2\(\frac{\mu_{2k_0}}{\de_{2k_0}}+\mu_0\),\qquad
		b:=\pm\(\frac{2\nu_{2k_0}}{\de_{2k_0}}+\ga\om k_0\),\\
		c:=\pm\(\frac{2\nu_{2k_0}}{\de_{2k_0}}-\ga\om k_0\),\qquad
		d:=\pm2\(-\frac{\mu_{2k_0}}{\de_{2k_0}}+\mu_0\),\qquad
		\upsilon:=\frac{2(1-\de_{2k_0})}{\de_{2k_0}}
	\end{gathered}
\end{equation}
with the same fixed sign in all $a$, $b$, $c$, $d$ depending on the sign of unscaled $\ep$, this equation is equivalent to the system
\begin{equation}\label{eqTh4_2}
	\begin{aligned}
		\wt{a}x+\wt{b}y & +x(x2+y^2)=0\\
		\wt{c}x+\wt{d}y & +y(x2+y^2)=0
	\end{aligned}
\end{equation}
for $\wt{a}=a/\upsilon$, $\wt{b}=b/\upsilon$, $\wt{c}=c/\upsilon$, $\wt{d}=d/\upsilon$. Note $\upsilon\neq 0$, i.e., $\de_{2k_0}\neq 1$. It is easy to see that $y=0$ whenever $x=0$ and vice versa. Hence we assume $x\neq 0\neq y$. Dividing the first equation by $x$, the second equation by $y$, and comparing the resulting equations, one derives
$$\wt{a}+\wt{b}\frac{y}{x}=\wt{c}\frac{x}{y}+\wt{d}.$$
Therefrom $x=q_{1,2}y$ for
$$q_{1,2}=\frac{(\wt{a}-\wt{d})\pm\sqrt{(\wt{a}-\wt{d})^2+4\wt{b}\wt{c}}}{2\wt{c}}$$
(here $q_1$ corresponds to ``$+$''-sign and $q_2$ to ``$-$''-sign).
As a consequence, putting $x$ in the second equation of \eqref{eqTh4_2} one obtains
$$y_{1,2}=\pm\sqrt{-\frac{\wt{c}q_{1,2}+\wt{d}}{1+q_{1,2}^2}}$$
(again $y_1$ corresponds to ``$+$'', $y_2$ to ``$-$''),
i.e., one has four points $(q_iy_j(q_i),y_j(q_i))$, $i,j=1,2$.

Condition \eqref{con2} means $ad-bc<0$, i.e., $\wt{a}\wt{d}-\wt{b}\wt{c}<0$. So if \eqref{con2} holds, then
$$\sqrt{(\wt{a}-\wt{d})^2+4\wt{b}\wt{c}}=\sqrt{(\wt{a}+\wt{d})^2-4(\wt{a}\wt{d}-\wt{b}\wt{c})}
	>|\wt{a}+\wt{d}|.$$
Accordingly,
\begin{equation}\label{eqTh4_3}
	\wt{c}q_{1,2}+\wt{d}=\frac{1}{2}\(\wt{a}+\wt{d}\pm\sqrt{(\wt{a}-\wt{d})^2+4\wt{b}\wt{c}}\)
\end{equation}
which is negative only for the minus sign (without any respect to the sign of unscaled $\ep$). Thus the nontrivial solutions of \eqref{eqTh4_1} are
\begin{equation}\label{eqTh4_p1}
	(q_2y_{1,2}(q_2),y_{1,2}(q_2)).
\end{equation}

On the other hand, assuming
\begin{equation}\label{con3}
	0<\de_{2k_0}^2(4\mu_0^2+\ga^2\om^2k_0^2)-4(\mu_{2k_0}^2+\nu_{2k_0}^2)<4\mu_0^2,
\end{equation}
equation \eqref{eqTh4_1} still has nontrivial solutions. Indeed, in this case
$$0<\sqrt{(\wt{a}-\wt{d})^2+4\wt{b}\wt{c}}=\sqrt{(\wt{a}+\wt{d})^2-4(\wt{a}\wt{d}-\wt{b}\wt{c})}
	<|\wt{a}+\wt{d}|,$$
so $q_{1,2}$ are both real, again. However, \eqref{eqTh4_3} is negative for both signs on suppose that $\wt{a}+\wt{d}<0$, i.e., $\pm\frac{\mu_0}{\upsilon}<0$ where the sign depends on the sign of unscaled $\ep$, e.g. if originally $\ep>0$ and $\frac{\mu_0}{\upsilon}<0$, equation \eqref{eqTh4_1} has four distinct nontrivial solutions
\begin{equation}\label{eqTh4_p2}
	(q_iy_j(q_i),y_j(q_i)),\qquad i,j=1,2,
\end{equation}
while for $\frac{\mu_0}{\upsilon}>0$, \eqref{eqTh4_1} has no nontrivial solutions.

In conclusion, if condition \eqref{con2} holds, there might exist two solutions of \eqref{eqR1} of order $O(\sqrt{|\ep|})$ for $\ep\neq 0$ small, and if \eqref{con3} is valid, there might exist four solutions of \eqref{eqR1} of order $O(\sqrt{|\ep|})$ for $\ep\sgn\frac{\mu_0\de_{2k_0}}{1-\de_{2k_0}}<0$ and no solutions of order $O(\sqrt{|\ep|})$ for $\ep\sgn\frac{\mu_0\de_{2k_0}}{1-\de_{2k_0}}>0$. We recall by Theorem \ref{Th4}, that there exists a small solution in all cases of lower order than $O(\sqrt{|\ep|})$, namely of order $O(\ep)$. So we might have 3, 5 and 1 small solutions, respectively. To apply the implicit function theorem and so to confirm the existence of these solutions, one has to show that the derivative of $\wt{H}_3$ at points \eqref{eqTh4_p1} or \eqref{eqTh4_p2} is surjective, which is rather awkward for general $a$, $b$, $c$, $d$, and so we do not go into details in this paper. But surely these solutions generically exist. Indeed, the determinant of the Jacobian of the left hand side of \eqref{eqTh4_2} at any $(qy,y)$ is
\begin{equation}\label{det1}
-\wt b\wt c+\wt a\wt d+\(3\wt a+\wt d+q(-2(\wt b+\wt c)+(\wt a+3\wt d)q)\)y^2+3\left(1+q^2\right)^2y^4.
\end{equation}
Now we fix all parameters only $\mu_0$ is variable. Then $q=q_{1,2}$, $\wt b$, $\wt c$ are constant (independent of $\mu_0$) and $y_{1,2}^2$, $\wt a$, $\wt d$ linearly depend on $\mu_0$: $\wt a=\pm\frac{2}{\upsilon}(r_1+\mu_0)$, $\wt d=\pm\frac{2}{\upsilon}(-r_1+\mu_0)$ and $y^2=y_{1,2}^2=\frac{-\wt c\upsilon q\pm2(r_1-\mu_0)}{\upsilon(1+q^2)}$ for constants $r_1$, $r_2$ independent of $\mu_0$. Inserting these expressions into \eqref{det1}, we have
\begin{equation}\label{det2}
\frac{1}{\upsilon^2(1+q^2)}\(a_\pm(r_1,q,\wt b,\wt c,\upsilon)+b_\pm(r_1,q,\wt b,\wt c,\upsilon)\mu_0\)
\end{equation}
for polynomials
$$
\begin{aligned}
a_+(r_1,q,\wt b,\wt c,\upsilon) &= \wt c\left(-\wt b+(\wt b+5\wt c)q^2+3\wt c q^4\right)\upsilon^2-4 q\left(\wt b+\wt c \left(5+2 q^2\right)\right) \upsilon r_1
	+16 r_1^2,\\
b_+(r_1,q,\wt b,\wt c,\upsilon) &= 4 \left(q \left(\wt b+\wt c \left(2+q^2\right)\right) \upsilon-4 r_1\right),\\
a_-(r_1,q,\wt b,\wt c,\upsilon) &= c \left(-\wt b+(\wt b+5\wt c) q^2+3\wt c q^4\right) \upsilon^2
	+4 r_1 \left(q \left(\wt b+c \left(5+2 q^2\right)\right) \upsilon+4 r_1\right),\\
b_-(r_1,q,\wt b,\wt c,\upsilon) &= - 4 \left(q \left(\wt b+\wt c \left(2+q^2\right)\right) \upsilon+4 r_1\right).
\end{aligned}$$
Now it is clear that $a_\pm$, $b_\pm$ are generically nonzero and then \eqref{det2} is nonzero for generic $\mu_0$. This shows that the above solutions exist in generic cases.

2. Let $p=\frac{\pi}{4}$, $\la=\frac{12}{37\sqrt{2}}$. Then by \eqref{eqRem1.1} we have
\begin{equation}\label{eqM2}
	M_2=\left\{\cos\frac{k\pi}{4}\mid k=\{0,1,\dots,7\}\right\}
		=\left\{0,\pm\frac{\sqrt{2}}{2},\pm{1}\right\}
\end{equation}
and condition \eqref{eqRem1.3} is fulfilled.
If $\om^2=\frac{37}{25}$, then $k_0=\pm 1$ satisfies the resonance condition \eqref{eqRem1.4}. Note that
$$\frac{1}{k^2(1-2\la\cos kp)}\leq\frac{1}{k^2(1-2\la)}=\frac{37}{k^2(37-12\sqrt{2})}
	<\frac{37}{13k^2}<\frac{37}{25}$$
whenever $|k|\geq 2$. Thus $k_0=\pm 1$ is the only resonant couple corresponding to $\om^2=\frac{37}{25}$. This is the case of a simple resonance discussed in this section.

On the other side, if $\om^2=\frac{37}{1225}$, then there are at least two couples $k_0=\pm 5, \pm 7$ such that \eqref{eqRem1.4} is fulfilled. Since
$$\frac{37}{k^2(37-12\sqrt{2})}<\frac{37}{20k^2}<\frac{37}{1225}$$
whenever $|k|\geq 8$, it is sufficient to evaluate $\frac{1}{k^2(1-2\la\cos kp)}$ at $k=\{1,\dots,7\}$ to determine the exact number of resonant modes. From Fig. \ref{fig1}, one can see that for $\om^2=\frac{37}{1225}$ there are two resonant couples $k_0=\pm 5,\pm 7$. This is a double resonance which can be investigated analogically to the present section.
\begin{figure}[htbp!]
\begin{center}
\includegraphics[width=0.5\textwidth]{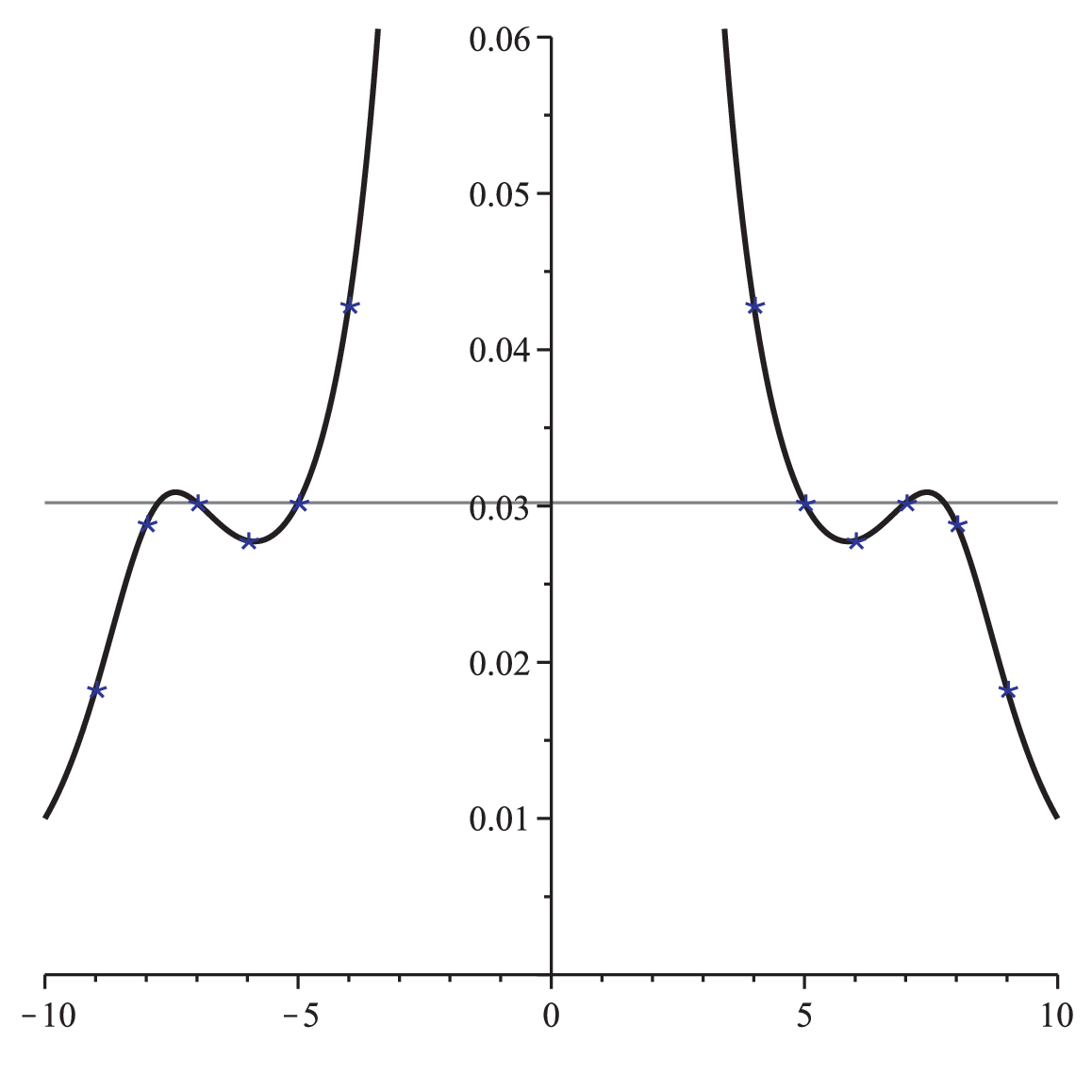}
\end{center}
\caption{The graph of functions $\frac{37}{1225}$ and $\frac{1}{x^2(1-2\la\cos xp)}$ with values at nonzero integers depicted with asterisks.}\label{fig1}
\end{figure}

\end{remark}

\section{Examples}\label{examples}

In this section we present several examples to illustrate the above results.

\begin{example}\label{ex1}
First we study the non-resonant case from Section \ref{per}. We consider the equation
\begin{equation}\label{eqex1}
    \begin{gathered}
		\om^2\(U(z)-\la U\(z-\frac{\pi}{4}\)-\la U\(z+\frac{\pi}{4}\)\)''+\ga\om U'(z)+U(z)\\
		{}+\om^2\(U^2(z)-\la U^2\(z-\frac{\pi}{4}\)-\la U^2\(z+\frac{\pi}{4}\)\)''+\ga\om(U^2(z))'-\ep\cos z=0
    \end{gathered}
\end{equation}
for $z\in\R$ and parameters $\ga,\la,\om>0$, $\frac{1}{2}\neq\la\neq\frac{\sqrt{2}}{2}$,
and apply Theorem \ref{Th1}. Note that equation \eqref{eqex1} is in the form of \eqref{eq2} with $p=\frac{\pi}{4}$, $h(z)=\ep\cos z$, and $M_2$ of \eqref{eqRem1.1} is equivalent to \eqref{eqM2}, i.e., condition \eqref{eqRem1.3} is satisfied. Hence using Remark \ref{Rem1}.3 and Theorem \ref{Th1}, we obtain the existence of a unique solution $U(\ep,z)=O(\ep)$ of equation \eqref{eqex1} for any $\ep$ sufficiently close to $0$. Since $U(\ep,z)$ is $C^\infty$--smooth, we can expand $U(\cdot,z)$ into Taylor series $U(\ep,z)=\ep \wt U_1(z)+O(\ep^2)$, and look for the function $\wt U_1$. It is easy to see that $\wt U_1$ is uniquely determined by the equation
$$\om^2\(\wt U_1(z)-\la\wt U_1\(z-\frac{\pi}{4}\)-\la\wt U_1\(z+\frac{\pi}{4}\)\)''+\ga\om\wt U_1'(z)+\wt U_1(z)-\cos z=0.$$
We look for its solution in the form of $\wt U_1(z)=A\cos z+B\sin z$, and, by comparing coefficients at $\cos z$ and $\sin z$, we obtain
\begin{equation*}
	-\ga\om A+(1+\om^2(\sqrt{2}\la-1))B = 0,\quad
	(1+\om^2(\sqrt{2}\la-1))A+\ga\om B = 1.
\end{equation*}
Therefrom,
$$A=\frac{1+\om^2(\sqrt{2}\la-1)}{\ga^2\om^2+(1+\om^2(\sqrt{2}\la-1))^2},\quad
	B=\frac{\ga\om}{\ga^2\om^2+(1+\om^2(\sqrt{2}\la-1))^2}.$$
In particular, if
\begin{equation}\label{par1}
\la=\frac{12}{37\sqrt{2}},\quad\om=\frac{\sqrt{37}}{5},
\end{equation}
then $A=0$, $B=\frac{5\sqrt{37}}{37\ga}$, i.e.,
\begin{equation}
U(\ep,z)=\frac{5\sqrt{37}}{37\ga}\ep\sin z+O(\ep^2).
\label{exp1}
\end{equation}
\end{example}

\begin{example}\label{ex2}
Now, moving onto resonant cases represented by equation \eqref{eqR1}, we consider the equation
\begin{equation}\label{eqex2}
    \begin{gathered}
		\om^2\(U(z)-\la U\(z-\frac{\pi}{4}\)-\la U\(z+\frac{\pi}{4}\)\)''+\ep\ga\om U'(z)+U(z)\\
		{}+\om^2\(U^2(z)-\la U^2\(z-\frac{\pi}{4}\)-\la U^2\(z+\frac{\pi}{4}\)\)''+\ep\ga\om(U^2(z))'-\ep\cos z=0,
    \end{gathered}
\end{equation}
i.e.\ when $\gamma$ in \eqref{eqex1} is replaced by $\ep\gamma$. Proceeding with point 2 of Remark \ref{Rem2} for parameters \eqref{par1}, then we obtain that $k_0=\pm 1$, $\de_{2k_0}=-\frac{123}{25}$ and $\upsilon=\frac{2(1-\de_{2k_0})}{\de_{2k_0}}=-\frac{296}{123}$.

Since $h(z)=\cos z=\frac{1}{2}\eu^{\imath z}+\frac{1}{2}\eu^{-\imath z}$ in \eqref{eqR2}, then $\mu_{k_0}=\frac{1}{2}$ and $\nu_{k_0}=0$. So Theorem \ref{Th2} can be applied to yield
$$
(x(0),y(0))=-\sqrt[3]{\frac{\de_{2k_0}}{2(1-\de_{2k_0})(\mu_{k_0}^2+\nu_{k_0}^2)}}(\mu_{k_0},\nu_{k_0})
=\left(\sqrt[3]{\frac{123}{74}}\frac{1}{2},0\right)\doteq(0.592281,0).
$$
Using \eqref{eqR5b}, \eqref{eqR5c} and \eqref{eqScala}, we derive that the solution from Theorem \ref{Th2} of \eqref{eqex2} with parameters \eqref{par1} is given by
\begin{equation}
\begin{gathered}
\sqrt[3]{\ep}c(\sqrt[3]{\ep})\eu^{\imath z}+\sqrt[3]{\ep}\bar c(\sqrt[3]{\ep})\eu^{-\imath z}+U_2\(\sqrt[3]{\ep}c(\sqrt[3]{\ep})\eu^{\imath z}+\sqrt[3]{\ep}\bar c(\sqrt[3]{\ep})\eu^{-\imath z},\ep\)\\
=\sqrt[3]{\ep}\(c(0)\eu^{\imath z}+\bar c(0)\eu^{-\imath z}+O\(\sqrt[3]{\ep}\)\)=\sqrt[3]{\ep}\(2x(0)\cos z+O\(\sqrt[3]{\ep}\)\)\\
=\sqrt[3]{\frac{123\ep}{74}}\cos z+O\(\sqrt[3]{\ep^2}\)\doteq \sqrt[3]{\ep}1.18456\cos z+O\(\sqrt[3]{\ep^2}\)
\end{gathered}
\label{exp2}
\end{equation}
for a smooth function $c(\cdot)$ with $c(0)=x(0)+y(0)\imath$ and $\ep\ne0$ small.
\end{example}

\begin{example}\label{ex3}
Similarly, Theorem \ref{Th3} can be applied in the equation
\begin{equation}\label{eqex3}
    \begin{gathered}
		\om^2\(U(z)-\la U\(z-\frac{\pi}{4}\)-\la U\(z+\frac{\pi}{4}\)\)''+\ep\ga\om U'(z)+U(z)\\
		{}+\om^2\(U^2(z)-\la U^2\(z-\frac{\pi}{4}\)-\la U^2\(z+\frac{\pi}{4}\)\)''+\ep\ga\om(U^2(z))'-\ep^2\cos z=0,
    \end{gathered}
\end{equation}
with
\begin{equation}\label{par2}
\la=\frac{12}{37\sqrt{2}},\quad\om=\frac{\sqrt{37}}{5},\quad \ga=0.1.
\end{equation}
Note that $\cos z$-term in \eqref{eqex2} has been replaced by $\ep\cos z$ in our case here, i.e., $h(z)$ in \eqref{eqR2} is replaced by $\ep h(z)$. Then using the proof of Theorem \ref{Th3}, we obtain that $(x(0),y(0))=\(0,-\frac{25}{\sqrt{37}}k_0\)$, since $k_0=\pm1$. Using \eqref{eqR5b}, \eqref{eqR5c} and \eqref{eqScalb}, we derive that the solution from Theorem \ref{Th3} of \eqref{eqex3} with parameters \eqref{par2} is given by
\begin{equation}
\begin{gathered}
	\ep c(\ep)\eu^{\imath k_0z}+\ep\bar{c}(\ep)\eu^{-\imath k_0z}
		+U_2(\ep c(\ep)\eu^{\imath k_0z}+\ep\bar{c}(\ep)\eu^{-\imath k_0z},\ep)\\
	=\ep(c(0)\eu^{\imath k_0z}+\bar{c}(0)\eu^{-\imath k_0z}+O(\ep))
		=\ep(-2y(0)\sin k_0z+O(\ep))\\
	=\frac{50\ep}{\sqrt{37}}k_0\sin k_0z+O(\ep^2)\doteq 8.21995\ep\sin z+O(\ep^2).
\end{gathered}
\label{exp3}
\end{equation}
\end{example}

\begin{example}\label{ex4}
Next, we consider
\begin{equation}\label{eqex4}
    \begin{gathered}
		\om^2\(U(z)-\la U\(z-\frac{\pi}{4}\)-\la U\(z+\frac{\pi}{4}\)\)''+\ep\ga\om U'(z)+U(z)\\
		{}+\om^2\(U^2(z)-\la U^2\(z-\frac{\pi}{4}\)-\la U^2\(z+\frac{\pi}{4}\)\)''+\ep\ga\om(U^2(z))'
         -\ep(\mu_0+2 (\cos 2z-\sin 2z))=0
    \end{gathered}
\end{equation}
with parameters \eqref{par2}, i.e., equation \eqref{eqR1} with the particular driving function $h(z)=\mu_0+2 (\cos 2z-\sin 2z)$. So, $\mu_{k_0}=\nu_{k_0}=0$ and $\mu_{2k_0}=\nu_{2}=-\nu_{-2}=1$. Hence, Theorems \ref{Th2} and \ref{Th3} cannot be used. However, using Theorem \ref{Th4} one can obtain that \eqref{con1} is of the form
$$
-8+\frac{15129}{625} \left(\frac{37}{2500}+4\mu_0^2\right)\ne0,
$$
i.e.\ $|\mu_0|\ne \frac{\sqrt{11940227}}{12300}\doteq 0.280932$. Moreover, \eqref{con2} is equivalent to
\begin{equation}\label{mu0_1}
	|\mu_0|<\frac{\sqrt{11940227}}{12300}.
\end{equation}
First, let $k_0=1$ and consider $\ep>0$, i.e.\ use ``$+$''-sign in \eqref{abcd}. Then we can derive that $q_2\doteq -0.294153$ and $y_{1,2}=\pm\sqrt{0.214884 + 0.764897\mu_0}$.
Note that $0<0.214884 + 0.764897\mu_0\leq 0.429768$. Using formula \eqref{det1}, we can estimate the Jacobian of $\wt{H}_3(q_2y_{1,2}(q_2),y_{1,2}(q_2),0)$ as
\begin{equation}\label{R2_est1}
	0<(0.218047+0.776154\mu_0)v^2\leq 2.52553,
\end{equation}
since the formula \eqref{det1} is equal to $\frac{1}{v^2} \det\du\wt{H}_3(qy,y,0)$. Therefore, we have 3 different small solutions for $k_0=1$, $\ep>0$ small. Using $k_0=-1$, we obtain $q_2\doteq 0.294153$, $y_{1,2}=\pm\sqrt{0.214884 + 0.764897\mu_0}$ and by \eqref{det1}, estimate \eqref{R2_est1} follows.
Hence for $k_0=\pm 1$, $\ep>0$ small and $\mu_0$ satisfying \eqref{mu0_1}, by \eqref{eqTh4_p1} we obtain nontrivial solutions of \eqref{eqTh4_1}. Scaling backwards in \eqref{eqScalc}, we obtain the following small solutions of \eqref{eqex4} with parameters \eqref{par2}
\begin{equation}
\begin{gathered}
	\sqrt{\ep} c(\sqrt{\ep})\eu^{\imath k_0z}
		+\sqrt{\ep}\bar{c}(\sqrt{\ep})\eu^{-\imath k_0z}
		+U_2(\sqrt{\ep} c(\sqrt{\ep})\eu^{\imath k_0z}
		+\sqrt{\ep}\bar{c}(\sqrt{\ep})\eu^{-\imath k_0z},\ep)\\
	=\sqrt{\ep}(2x(0)\cos k_0z-2y(0)\sin k_0z)+O(\ep)
		=2\sqrt{\ep}y(0)(q_2\cos k_0z-\sin k_0z)+O(\ep)\\
	=\pm 2\sqrt{\ep(0.214884 + 0.764897\mu_0)}(0.294153\cos z+\sin z)+O(\ep).
\end{gathered}
\label{exp4}
\end{equation}
Moreover, there is another (smallest) solution of \eqref{eqex4} with parameters \eqref{par2} of order $U(z)=O(\ep)$. Using \eqref{der}, we obtain $x'(0)=y'(0)=0$ and by \eqref{sol1}, the solution is
\begin{equation}\label{eqUsmallest}
	U(\ep)=\ep\(\mu_0-\frac{50}{123}\cos 2z+\frac{50}{123}\sin 2z\)+O(\ep^2).
\end{equation}

Considering $k_0=1$ and $\ep<0$, i.e.\ using ``$-$''-sign in \eqref{abcd}, we derive $q_2\doteq 1.83348$ and $y_{1,2}=\pm\sqrt{0.0535297 - 0.190543\mu_0}$. Note that $0<0.0535297 - 0.190543\mu_0\leq 0.107059$. Formula \eqref{det1} yields
\begin{equation}\label{R2_est2}
	0<0.218047-0.776154\mu_0\leq 0.436093,
\end{equation}
which implies that we also have 3 different small solutions for $k_0=1$, $\ep<0$ small and $\mu_0$ satisfying \eqref{mu0_1}. Note that if $k_0=-1$, $\ep<0$, we obtain $q_2\doteq -1.83348$, $y_{1,2}=\pm\sqrt{0.0535297 - 0.190543\mu_0}$ and, by \eqref{det1}, estimate \eqref{R2_est2}. The 3 solutions of \eqref{eqex4} with parameters \eqref{par2} have the forms \eqref{eqUsmallest} and
$$
\begin{gathered}
	\sqrt{-\ep} c(\sqrt{-\ep})\eu^{\imath k_0z}
		+\sqrt{-\ep}\bar{c}(\sqrt{-\ep})\eu^{-\imath k_0z}
	+U_2\(\sqrt{-\ep} c(\sqrt{-\ep})\eu^{\imath k_0z}
		+\sqrt{-\ep}\bar{c}(\sqrt{-\ep})\eu^{-\imath k_0z},\ep\)\\
	=\sqrt{-\ep}(2x(0)\cos k_0z-2y(0)\sin k_0z)+O(\ep)
	=2\sqrt{-\ep}y(0)(q_2\cos k_0z-\sin k_0z)+O(\ep)\\
	=\pm 2\sqrt{-\ep(0.0535297 - 0.190543\mu_0)}(1.83348\cos z-\sin z)+O(\ep).
\end{gathered}
$$

Next, \eqref{con3} implies
\begin{equation}\label{mu0_2}
	0.280932\doteq\frac{\sqrt{11940227}}{12300}<|\mu_0|
		<\frac{\sqrt{\frac{11940227}{74}}}{1400}\doteq0.286921.
\end{equation}
From the above computations, we can verify that formula \eqref{det1} has one of the forms $0.218047\pm0.776154\mu_0\ne0$. Since $\sgn\upsilon=-1$, we have that: if $\ep$ is small, $\ep\sgn \mu_0>0$ and $\mu_0$ fulfills \eqref{mu0_2}, then there are 5 different small solutions ($4\times O(\sqrt{|\ep|})$ and $1\times O(\ep)$), if $\ep\sgn \mu_0<0$ there is only 1 small solution of order $O(\sqrt{|\ep|})$, actually \eqref{eqUsmallest}. For instance, if $\ep>0$ and $\mu_0$ satisfies \eqref{mu0_2}, we obtain
\renewcommand\arraystretch{1.5}
$$
\begin{array}{c|c|c}
k_0 & q_i & y_i\\\hline
1 & q_1=1.83348 & y_{1,2}(q_1)=\pm\sqrt{-0.0535297 + 0.190543\mu_0}\\
  & q_2=-0.294153 & y_{1,2}(q_2)=\pm\sqrt{0.214884 + 0.764897\mu_0}\\\hline
-1 & q_1=-1.83348 & y_{1,2}(q_1)=\pm\sqrt{-0.0535297 + 0.190543\mu_0}\\
  & q_2=0.294153 & y_{1,2}(q_2)=\pm\sqrt{0.214884 + 0.764897\mu_0}
\end{array}$$
Hence, equation \eqref{eqex4} with parameters \eqref{par2} has the solutions \eqref{eqUsmallest} and
\begin{gather*}
	\pm 2\sqrt{\ep(-0.0535297 + 0.190543\mu_0)}(1.83348\cos z-\sin z)+O(\ep),\\
	\pm 2\sqrt{\ep(0.214884 + 0.764897\mu_0)}(0.294153\cos z+\sin z)+O(\ep)
\end{gather*}
for $\mu_0>0$, and the only solution \eqref{eqUsmallest} of order $O(\sqrt{\ep})$ for $\mu_0<0$.
\end{example}

\section{Numerical results}\label{nums}

To illustrate the theoretical results obtained in the previous sections, we have solved the governing equations \eqref{eq1} and \eqref{eq2} numerically. The advance-delay equation \eqref{eq2} is solved using a pseudo-spectral method, i.e.\ we express the solution $U$ in the Fourier series
\begin{equation}
U(z)=\sum_{j=1}^{J}\left[B_j\cos\left((j-1)\tilde{k}z\right)+C_j\sin\left(j\tilde{k}z\right)\right],
\label{ser}
\end{equation}
where $\tilde{k}=2\pi/L$ and $-L/2<z<L/2$. Substituting the series in \eqref{eq2} and projecting it onto the Fourier modes in a reminiscence of the Galerkin truncation method, we obtain a system of algebraic equations for the Fourier coefficients $B_j$ and $C_j$, which are then solved using, e.g., a Newton-Raphson method. Typically, we use $L=2\pi$ and large $J$ that will be indicated in each calculation below.

The stability of a solution obtained from \eqref{eq2} is determined numerically through calculating its Floquet multipliers, which are eigenvalues of the monodromy matrix. As a first-order system, the linearized equation of \eqref{eq1} is
\begin{align}
\dot{u}_n &= v_n,\\
		\Delta[\dot{v}_n(1+2U_n)] &= -\ga v_n(1+2U_n)-u_n(1+2\ga \dot{U}_n)
		-2\Delta[\ddot{U}_nu_n+2\dot{U}_nv_n],
\end{align}
where $\Delta[\bullet_n]=\bullet_n-\la\bullet_{n-1}-\la\bullet_{n+1}$ and $U_n(t)=U(\omega t+np)$ (cf.\ \eqref{ser}) is a periodic solution of \eqref{eq1}. The linear system is integrated using a Runge-Kutta method of order four with periodic boundary conditions over the period $T=2\pi/\omega$ and with the site index $n=1,2,\dots,N$. The $i$th column of the monodromy matrix $M$ is vector $[u_1(T),\dots,u_N(T),v_1(T),\dots,v_N(T)]^\mathrm{t}$, that corresponds to the initial condition $[u_1(0),\dots,u_N(0),v_1(0),\dots,v_N(0)]^\mathrm{t}$ equal to the $i$th column vector of the identity matrix $I_{2N}$. A periodic solution is asymptotically stable if all Floquet multipliers except the trivial Floquet multiplier are strictly smaller than one in modulus. The (in)stability result obtained from calculating the monodromy matrix is also confirmed by integrating \eqref{eq1}. Note that the physically relevant interval for the coupling parameter is $|\la|<1/2$ \cite{dibl14}. It is because there will be a band of unstable multipliers when $|\la|>1/2$, implying that all solutions of the system will be modulationally unstable. Following Section \ref{examples}, we set $\lambda=12/\left(37\sqrt{2}\right)$ and $p=\pi/4$. Computations for different parameter values have been performed, where we did not see any significant difference.

\begin{figure}[tbhp!]
\begin{center}
\subfigure[]{\includegraphics[width=0.45\textwidth]{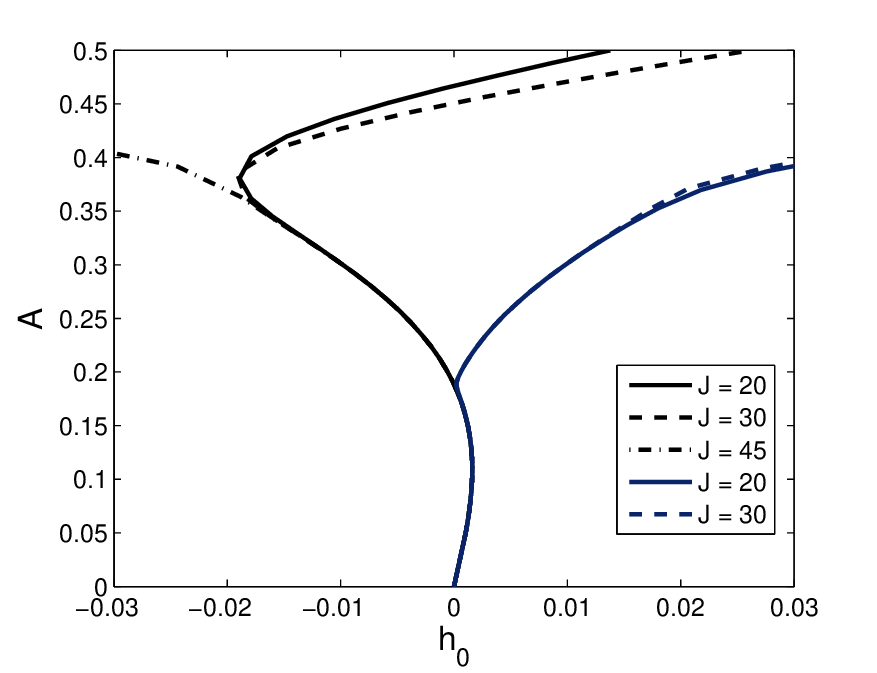}}
\subfigure[]{\includegraphics[width=0.45\textwidth]{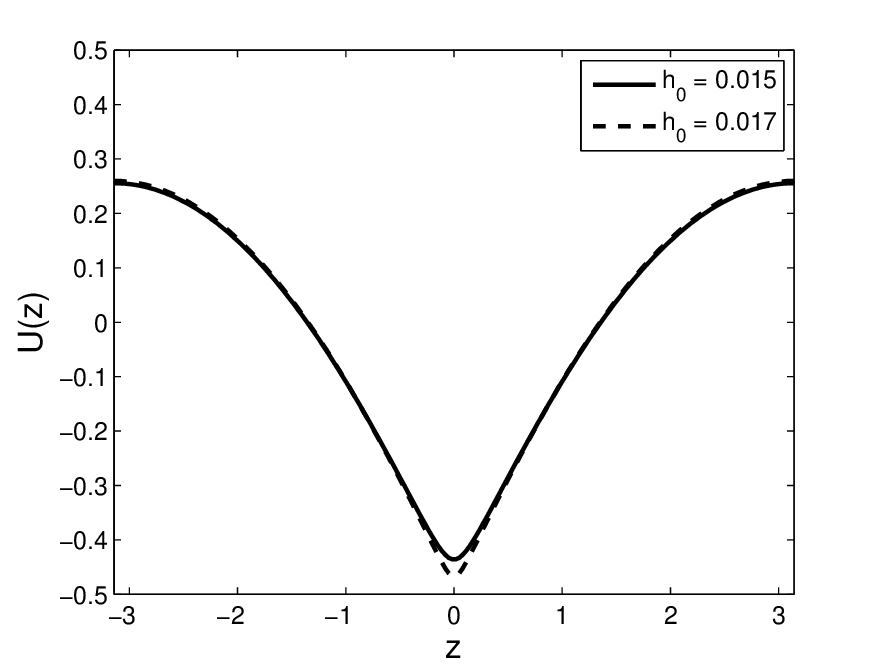}}
\end{center}
\caption{(a) The oscillation amplitude of periodic solutions of \eqref{eq2} as a function of the driving amplitude $h_0$ for $\omega=1.23$ with $\gamma = 0$ (black thick lines) and $\gamma=0.001$ (blue thin lines). (b) Solution profiles at the driving amplitudes indicated in the legend for vanishing $\gamma$.}\label{figw2}
\end{figure}

First, we consider the case of
\begin{equation}
h(z)=h_0\cos z.
\end{equation}
From \eqref{eqRem1.4} for $\gamma=0$ the first resonance of the system is $k=\pm1$, which corresponds to $\omega=\sqrt{37}/5$. In Fig.\ \ref{figw2}, we depict the bifurcation diagrams of periodic solutions for a value of $\omega=1.23>\sqrt{37}/5$. On the vertical axis is the oscillation amplitude $A$ of the periodic solutions defined as \[A=\frac12\left(\text{max}\,(U(z))-\text{min}\,(U(z))\right)\] as a function of the driving amplitude $h_0$.

To check the convergence and dependence of the method on the number of modes, we computed the bifurcation diagrams for several values of $J$. Interestingly we observed that there is always a critical value of $h_0$ beyond which the bifurcation diagrams depend significantly on the parameter. The diagrams for three values of $J$ are shown in Fig.\ \ref{figw2}(a). Using the figure, we conclude that to determine the validity region of our numerical method, we need to compute bifurcation diagrams using at least two different values of $J$. The 'breaking point' of the method is when the curves start branching out. Studying the solution profiles near the breaking point as shown in Fig.\ \ref{figw2}(b), the breaking point seemingly corresponds to an extreme point of $U(z)$ having discontinuous first derivative. We therefore infer that our equation \eqref{eq2} may have peaked-periodic solutions.

As shown as black thick lines in Fig.\ \ref{figw2}(a), when $\gamma=0$ there is one saddle-node bifurcation that occurs as we vary $h_0$. Note that the bifurcation diagram has a mirror symmetric (with respect to the vertical axis $h_0=0$) counterpart, i.e.\ the system is symmetric with the transformation $h_0\to-h_0,\,U\to -U$. When $\gamma$ is turned on to a non-zero value, the bifurcation curve interacts with its symmetry and at the intersection point $h_0=0$ for non-vanishing $A$ an additional turning-point is formed, i.e.\ there is a curve-splitting-and-merging. In that case, as shown as blue thin lines in Fig.\ \ref{figw2}(b), there are now two saddle-node bifurcations.

\begin{figure}[tbhp!]
\begin{center}
\subfigure[]{\includegraphics[width=0.45\textwidth]{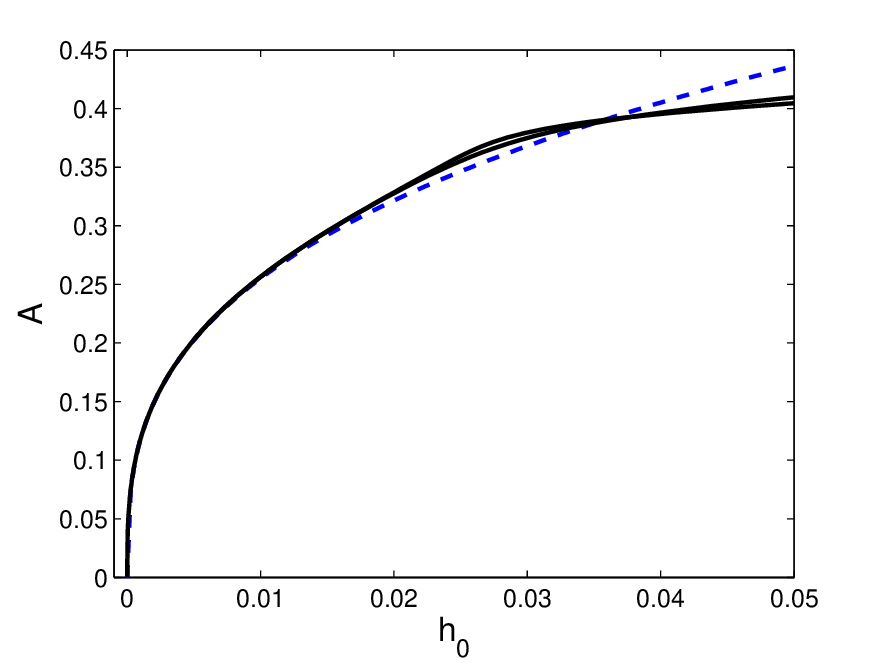}}
\subfigure[]{\includegraphics[width=0.45\textwidth]{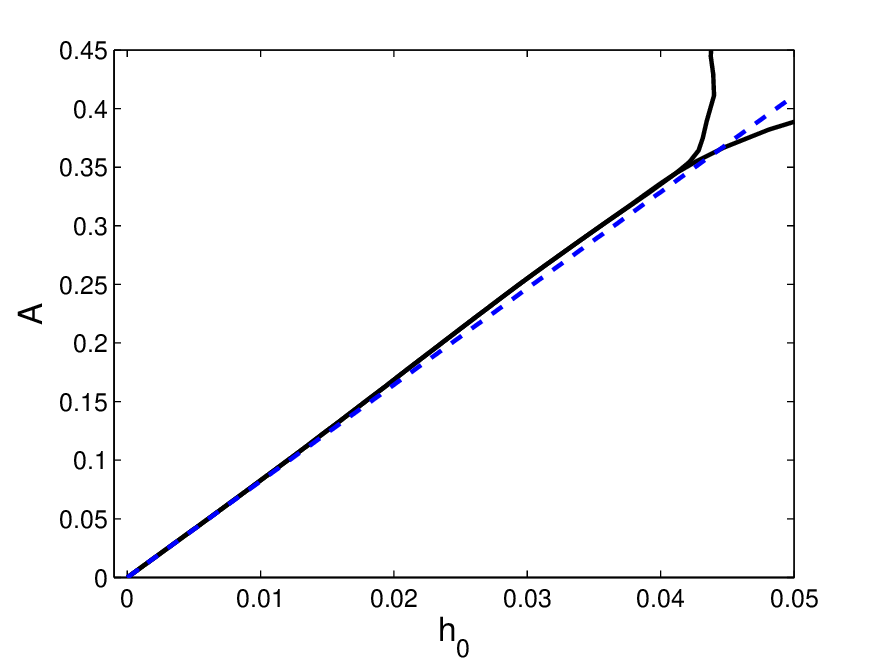}}
\end{center}
\caption{The same as Fig.\ \ref{figw2}(a) for $\omega=\sqrt{37}/5$ with (a) $\gamma=0$ and (b) $\gamma=0.1$. In both panels, the solid black and the dashed blue lines are the numerics and the analytical results \eqref{exp2} (a) and \eqref{exp1}, \eqref{exp3} (b), respectively.}
\label{fig3}
\end{figure}

When $\omega$ is decreased toward the resonance frequency, the critical drive of $h_0$ for the occurrence of the first saddle-node bifurcation, i.e.\ the bifurcation that also exists when $\gamma=0$, decreases towards zero. In the limit $\omega=\sqrt{37}/5$, we present in Fig. \ref{fig3}(a) the bifurcation curve of periodic solutions at the resonance. The numerics (shown as solid black curves) is in good agreement with the analytical result (dashed line) in Example \ref{ex2}, i.e.\ \eqref{exp2}, that when $\gamma=0$ the slope of the existence curve is singular at $h_0=0$.

In panel (b) of the same figure, we depict the bifurcation curve of the periodic solutions for the same parameter values, but with $\ga\neq0$. In agreement with the analytical result in Examples \ref{ex1} and \ref{ex3}, i.e.\ \eqref{exp1} and \eqref{exp3}, adding dissipation will regularize the slope at the origin.

\begin{figure}[tbhp!]
\begin{center}
\subfigure[]{\includegraphics[width=0.45\textwidth]{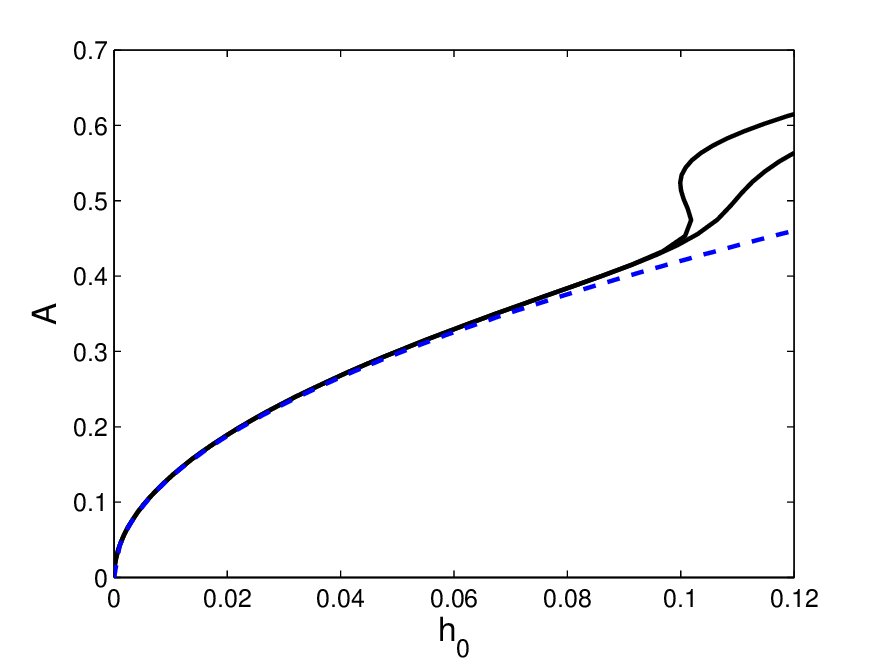}}
\subfigure[]{\includegraphics[width=0.45\textwidth]{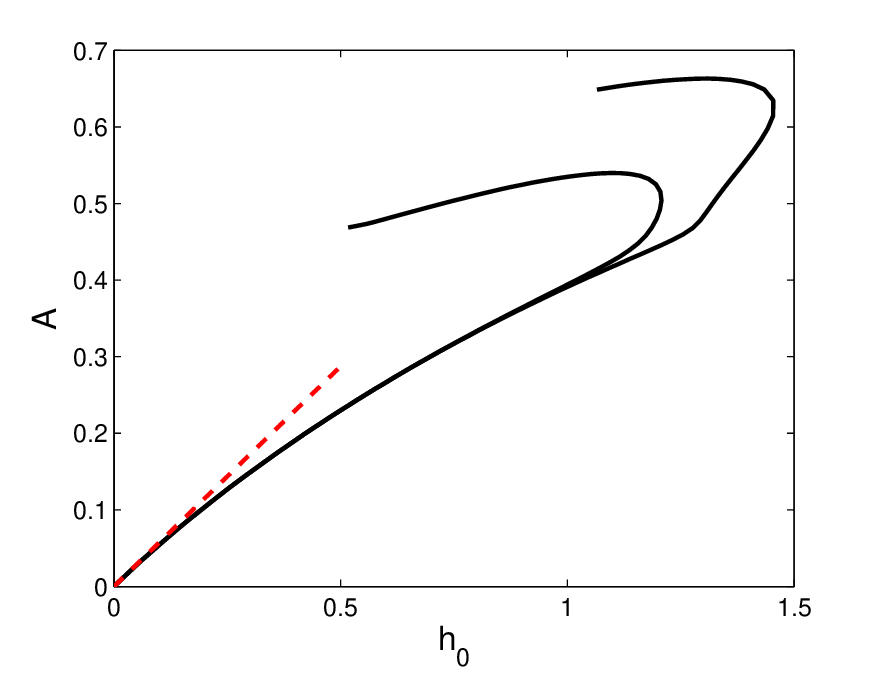}}
\end{center}
\caption{The same as Fig.\ \ref{figw2}(a), but for the periodic drive \eqref{h2}. The dashed lines are the analytical amplitudes from \eqref{exp4} and \eqref{eqUsmallest}. Here, $\gamma=0$.}
\label{fig4}
\end{figure}

\begin{figure}[tbhp!]
\begin{center}
(a)\includegraphics[width=0.52\textwidth]{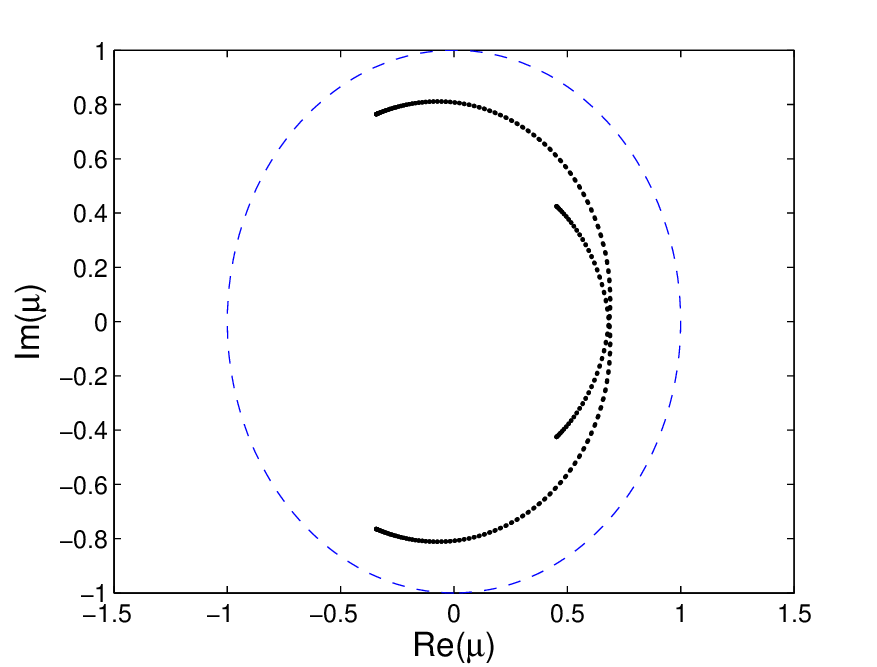}
(b)\includegraphics[width=0.52\textwidth]{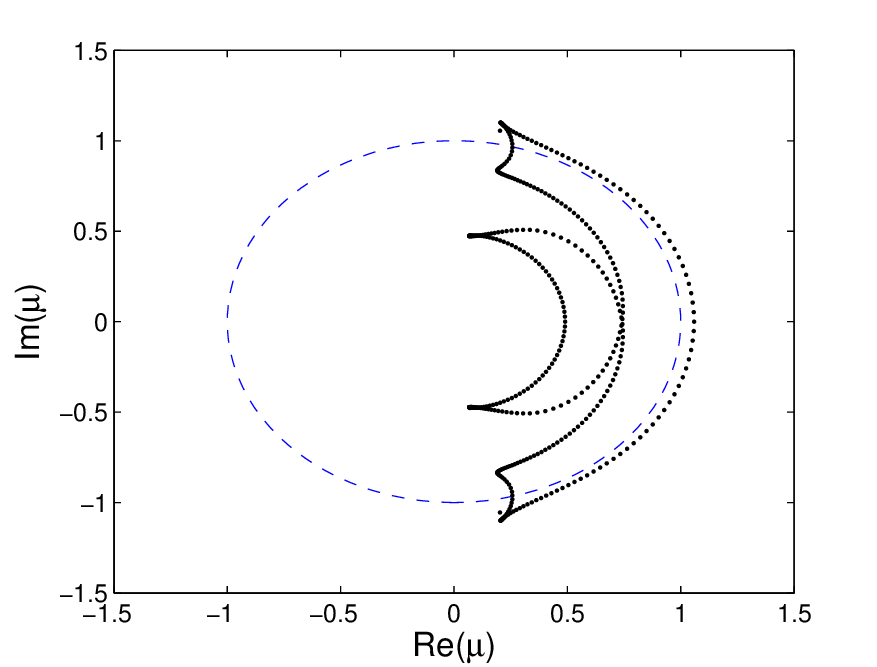}
(c)\includegraphics[width=0.52\textwidth]{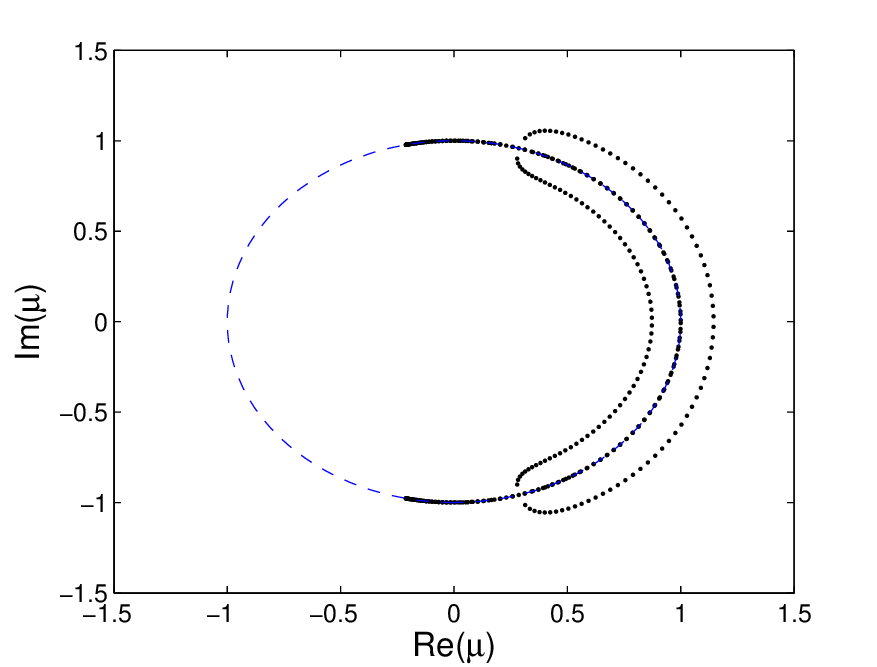}
\end{center}
\caption{Floquet multipliers of periodic solutions for some parameter values (see the text). The dashed curve is the unit circle, showing as guide to the eye.}
\label{fig5}
\end{figure}

We have also considered the multiple resonance case discussed in Section \ref{examples}. For the numerics, we take the drive
\begin{equation}
h(z)=h_0\(\mu_0+2(\cos(2 z)-\sin(2z))\),
\label{h2}
\end{equation}
with $\mu_0=1/4$, which satisfies the condition \eqref{mu0_1}. From Example \ref{ex4}, we know that in this resonance case, there are three small solutions given by \eqref{exp4} and \eqref{eqUsmallest}. Shown in Fig.\ \ref{fig4} are numerically computed bifurcation curves of two of the three periodic solutions.

In panel (a), we show only one bifurcation diagram of the solutions \eqref{exp4}. It is because the amplitudes of the two solutions are almost the same. In panel (b), we depict the smaller solution with amplitudes that scale as $O(h_0)$. One can note that in both plots our analytical results from Lyapunov-Schmidt reduction agree with the numerics.

We have also determined the stability of the solutions obtained numerically and found that in the non-resonant condition, the periodic solutions are generally stable for small enough driving amplitude $h_0$.

Shown in Fig.\ \ref{fig5} are the Floquet multipliers of some periodic solutions, calculated using $N=200$. Panels (a-b) depict the multipliers of the periodic solutions corresponding to the bifurcation curve in Fig.\ \ref{fig3}(b) for $h_0=0.01$ and $0.04$, respectively. As all the multipliers in the first panel are inside the unit circle, we can deduce that the solution is stable. However, as $h_0$ increases further, some multipliers leave the circle as shown in panel (b) and induce modulational instability. Panel (c) of the same figure corresponds to a solution along the existence curve in Fig.\ \ref{fig3}(a), i.e.\ with $h_0=0.01$ and $\ga=0$. It is clear that the solution is unstable. For this resonant case, we found that all solutions for any driving amplitude are modulationally unstable.

We have also simulated the dynamics of the unstable solutions by solving the governing equation \eqref{eq1}. We observed that the typical dynamics of the instability is in the form of solution blow-up, similarly to that in \cite{agao14}.

As a conclusion, we found that the qualitative behaviour of the system \eqref{eq1} is similar to that of forced FPU lattices discussed in \cite{feck13}, despite the significant difference in the coupling terms as explained in Section \ref{intro}.

\section*{Acknowledgments}

%

The work of MA has been co-financed from resources of the operational program ”Education and Lifelong Learning” of the European Social Fund and the National Strategic Reference Framework (NSRF)
2007-2013 within the framework of the Action State Scholarships Foundation’s (IKY) mobility grants programme for the short term training in recognized scientific/research centres abroad for candidate doctoral or postdoctoral researchers in Greek universities or research centres. MF is partially supported by Grants VEGA-MS 1/0071/14, VEGA-SAV 2/0029/13 and by the Slovak Research and Development Agency under the contract  No. APVV-14-0378. MP is supported by Grant VEGA-SAV 2/0029/13. The work of VR has been co-financed by the European Union (European Social Fund ESF) and
Greek national funds through the Operational Program Education and Lifelong Learning of the National Strategic Reference Framework (NSRF) Research Funding Program: THALES Investing in knowledge society through the European Social Fund. VR and HS acknowledge partial support from the London Mathematical Society through a visitor grant.






\bibliographystyle{elsarticle-num}







\end{document}